\documentclass[11pt,twoside]{article}

\usepackage[ansinew]{inputenc} 
\usepackage{amsmath,amsfonts,amssymb,amsthm}
\usepackage{pstricks,pst-node,pst-coil,pst-plot,pstricks-add}
\usepackage{graphics,geometry,epsfig}
\usepackage{bm,bbm}

\usepackage{mathrsfs} 

\usepackage[pdftex]{hyperref}

\usepackage{graphicx}

\usepackage{units}
\usepackage[squaren]{SIunits}

\usepackage{tabularx,dcolumn,booktabs}

\usepackage{rotate}

\usepackage{float}

\usepackage{dsfont}

\usepackage{chngcntr}

\usepackage{icomma}

\usepackage{dcolumn}
\usepackage{multirow}
\usepackage{booktabs}

\usepackage[width = 0.9\linewidth, textfont = {small}, labelfont = {small, bf}, format = plain]{caption}

\usepackage[innercaption]{sidecap}

\newcommand{\field}[1]{\mathbb{#1}}

\newcommand{\R}{\field{R}}
\newcommand{\C}{\field{C}}

\newcommand{\DD}{\mathscr D}

\newcommand{\HH}{\mathscr H}

\newcommand{\FF}{\mathcal F}

\newcommand{\ID}{\mathds{1}} 

\newcommand{\Ima}{\operatorname{Im}}
\newcommand{\Rea}{\operatorname{Re}}

\newcommand{\limop}[2][\infty]{\lim\limits_{#2 \rightarrow #1}}

\newcommand{\liminfop}[2][\infty]{\liminf\limits_{#2 \rightarrow #1}}
\newcommand{\limsupop}[2][\infty]{\limsup\limits_{#2 \rightarrow #1}}

\newcommand{\f}{v}
\newcommand{\chibar}{\overline{\chi}}
\newcommand{\eps}{\varepsilon}

\newcommand{\abs}[1]{\mbox{$\left| #1 \right| $}}                 
\newcommand{\norm}[1]{\mbox{$\left\| #1 \right\|$}}           
\newcommand{\sprod}[2]{\mbox{$\left\langle #1,#2 \right\rangle$}}        
\newcommand{\fock}{\FF}
\newcommand{\hilbert}{\HH}
\newcommand{\uv}{\Lambda}
\newcommand{\ir}{K}
\newcommand{\irFix}{\sigma}
\newcommand{\bodisp}{\omega}

\newcommand{\assDisp}{(\bodisp)}
\newcommand{\formFac}{(\f 1)}
\newcommand{\regTwo}{(\f 3)}
\newcommand{\regThr}{(\f 2)}

\newtheorem{theorem}{Theorem}[section]
\newtheorem{lemma}[theorem]{Lemma}
\newtheorem{corollary}[theorem]{Corollary}

\theoremstyle{plain}

\title{On the domain of the Nelson Hamiltonian}

\author{M. Griesemer and A. W\"unsch\\  
\small Fachbereich Mathematik, Universit\"at Stuttgart, D-70569 Stuttgart, Germany}  

\date{}

\begin{document}
\maketitle

\begin{abstract}
The Nelson Hamiltonian is unitarily equivalent to a Hamiltonian defined through a closed, semibounded quadratic form, the unitary transformation being explicitly known and due to Gross. In this paper we study mapping properties of the Gross-transform in order to characterize regularity properties of vectors in the form domain of the Nelson Hamiltonian. Since the operator domain is a subset of the form domain, our results apply to vectors in the domain of the Hamiltonian was well. - This work is a continuation of our previous work on the Fr\"ohlich Hamiltonian. 
\end{abstract}



\section{Introduction}
\label{sec:nelsonIntroduction}

There is well known model, due to Nelson, describing a system of $N$ non-relativistic quantum particles (nucleons) interacting with a quantized field of scalar bosons (pions).  The Hamiltonian of this model is based on a formal expression, which, in the case $N=1$, is given by
\begin{equation}\label{formal-Nelson}  
     -\Delta + d\Gamma(\bodisp) +\int \f(k)\left(e^{ikx}a_k +  e^{-ikx}a^{*}_k\right)\, dk
\end{equation}
where $\Delta$ denotes the Laplacian on $L^2(\R^3)$, $d\Gamma(\bodisp) = \int \bodisp(k) a^{*}_k a_k\, dk$ measures the field energy, $\bodisp(k) = \sqrt{k^2+m^2}$ with $m\geq 0$ and $\f(k)=\bodisp(k)^{-1/2}$. The integral in \eqref{formal-Nelson} accounts for the particle field interaction. Due to the non-square integrable decay of the form-factor $\f$, the expression \eqref{formal-Nelson}, as it stands, does not define a densely defined self-adjoint operator and its quadratic form is unbounded from below. There is a well-known  procedure to cure these problems: upon introducing an ultraviolet cutoff $|k|\leq \uv$ in the interaction integral, the resulting Hamiltonian $H_\uv$ conjugated with a dressing transform $U_\uv$ reveals a divergent vacuum expectation energy $-E_\uv$. When this energy is subtracted, the regularised, dressed Hamiltonian, in the limit $\uv\to \infty$ defines a closed and semi-bounded  quadratic form and hence a self-adjoint Hamiltonian. By reversing the dressing, now using $U_{\infty}$, a self-adjoint Hamiltonian, the Nelson-Hamiltonian, is obtained \cite{Nelson1964}. By general arguments this Hamiltonian is the limit of $H_\uv+E_\uv$ as $\uv\to\infty$ in the norm-resolvent sense.  

The procedure described above does not provide an explicit expression for the Nelson Hamiltonian, let alone its domain. It does provide, however, the form domain $D(q)$, which is given by
\begin{equation}\label{form-domain}
     D(q) : = U_{\infty}^{*} \left(D(|p|)\cap D(d\Gamma(\bodisp)^{1/2})\right),
\end{equation}
where $|p|=\sqrt{-\Delta}$. In this paper we study the mapping properties of $U_{\infty}^{*}$. Our main results imply that
\begin{align}
\label{equ:introResultDim3Inclus}   D(q)&\subset\bigcap\limits_{0\leq s < 1} D(\abs{p}^s)\cap D(d\Gamma(\bodisp)^{1/2}), \\
\label{equ:introResultDim3Intersec}   &D(q)\cap D(\abs{p})=\{0\}.
\end{align}
Since $D(H)\subset D(q)$,  statements \eqref{equ:introResultDim3Inclus} and \eqref{equ:introResultDim3Intersec} hold equally for $D(H)$, and since $D(H_0^{1/2})\subset D(|p|)$ for $H_0:=-\Delta+d\Gamma(\omega)$, we conclude from \eqref{equ:introResultDim3Intersec} that 
\begin{align*}
   D(H)\cap D(H_0^{1/2})=\{0\}.
\end{align*}
This means in particular that $D(H)\cap D(H_0)=\{0\}$.

On a heuristic level the critical exponent $s=1$ of $|p|$ in \eqref{equ:introResultDim3Inclus} and \eqref{equ:introResultDim3Intersec}  can easily be understood on the basis of \eqref{form-domain}, where $U_{\infty}^{*} = \exp(a^{*}(B_x)-a(B_x))$ and 
\begin{equation*}
      B_x(k) = \frac{\f(k)}{\bodisp(k)+k^2} e^{ikx} \chi(|k|\geq \ir) \qquad (\ir>0).
\end{equation*}
For $\Psi\in  D(|p|)\cap D(d\Gamma(\bodisp)^{1/2})$ and $s\leq 1$ we expect that $U_{\infty}^{*} \Psi\in D(|p|^s)$ if and only if the norm of the operator $[|p|^s,B_x(k)]$ is square integrable as a function of $k$. This is true if and only if $|k|^sB_x(k)$ is square integrable, which, for the Nelson model in three dimensions, is satisfied if and only if $s<1$. The problem solved in this paper is to make this heuristic argument rigorous, the hard part being the proof of \eqref{equ:introResultDim3Intersec}.

In \cite{GriesemerWuensch} we had established results similar to \eqref{equ:introResultDim3Inclus} and \eqref{equ:introResultDim3Intersec} for the Fr\"ohlich Hamiltonian, which corresponds to \eqref{formal-Nelson} with $\bodisp(k)\equiv 1$ and $\f(k)=|k|^{-1}$. In this case  \eqref{formal-Nelson} defines a closed, semi-bounded quadratic form and hence the procedure of Nelson is not necessary for defining a Hamiltonian, but it can still be applied. It turns out that the dressed Hamiltonian is self-adjoint on $D(H_0)$ and the analysis of $ U_{\infty}^{*} D(H_0)$ is simplified by the fact that $\bodisp\equiv 1$. Similar remarks apply to the Nelson model in one and two space dimensions, see Section~\ref{sec:Nelson-low} and \cite{DissWuensch2017}. In the present paper we concentrate on the only open problem, which is the proof of \eqref{equ:introResultDim3Inclus} and \eqref{equ:introResultDim3Intersec} for a class of models, defined in terms of assumptions on $\bodisp$ and $\f$, that is taylor made for the Nelson  model in three dimensions with massive or massless bosons.

In Section~\ref{sec:nelsonAndGross} we describe the class of models to be considered in the ultraviolet regularized form and we introduce the corresponding class of dressing transforms. In Section~\ref{sec:construct} the construction of the Nelson Hamiltonian is given for the class of Section~\ref{sec:nelsonAndGross}, the abstract part of the argument being deferred to Appendix~\ref{app:AmariLikeTheorem}. Section~\ref{sec:domain}  is devoted to the mapping properties of $U_{\infty}^{*}$ and it contains our main results, Theorem~\ref{thm:domainSubset3D4D} and Theorem~\ref{thm:domainIntersec3D4D}. From these theorems we derive \eqref{equ:introResultDim3Inclus} and \eqref{equ:introResultDim3Intersec} for the Nelson model in Section~\ref{sec:Nelson-low}. Section~\ref{sec:Nelson-low} also describes the improved results that can be shown for the Nelson model in dimension $d\leq 2$. There are two appendices, besides Appendix~\ref{app:AmariLikeTheorem}, were tools for the proofs of  Section~\ref{sec:domain} are collected.


\section{Hamiltonian with cutoff and Gross transform}
\label{sec:nelsonAndGross}

In this section we fix our notations, we define the class of regularized Hamiltonians $H_\uv$, $\uv<\infty$, to be considered in this paper and we introduce our assumption on $\bodisp$ and $\f$.

Let $\hilbert:=L^2(\R^d,dx)\otimes\fock$ where $\fock$ denotes the symmetric Fock space over 
$L^2(\R^d,dk)$. Through the isomorphism defined by $\varphi\otimes\eta\mapsto\varphi(x)\eta$ we may 
identify $\hilbert$ with $L^2(\R^d,\fock)$, the space of square integrable functions $x\mapsto\Psi(x)\in \fock$ on $\R^d$. The Fourier transform of such a function will be denoted by $\Psi(p)$. As usual we use $\Psi^{(n)}$ to denote the $n$-boson component of the vector $\Psi$. 

With $\fock_{0}$ and $\hilbert_{0}$ we denote the subspaces 
\begin{align}
   \fock_{0}:=\bigcup\limits_{n\geq 0}\chi(N\leq n)\fock   \quad\text{and}\quad   \hilbert_{0}:=\bigcup\limits_{n\geq 0}\chi(N\leq n)\hilbert
\end{align}
of finite particle vectors in  $\fock$ and $\hilbert$, respectively. Here in all the following $N$ denotes the number operator, which is  defined by $(N\Psi)^{(n)} = n\Psi^{(n)}$.

The non-interacting system composed of particle and quantized field is described by the free Hamiltonian
\begin{align}
   H_0:=p^2\otimes\ID+\ID\otimes d\Gamma(\bodisp), 
\end{align}
on $\hilbert$, where  $p^2=-\Delta$ on $L^2(\R^d)$ and $d\Gamma(\bodisp)$ denotes the second quantization of the one-particle operator given by multiplication with the dispersion relation $\bodisp$. That is, $ (d\Gamma(\bodisp)\Psi)^{(0)}=0$ and for $n\geq 1$,
$$
    (d\Gamma(\bodisp)\Psi)^{(n)}(k_1,\ldots,k_n) = \sum_{i=1}^n\bodisp(k_i)\Psi^{(n)}(k_1,\ldots,k_n).
$$
On $\bodisp$ we will always assume that $\bodisp\in L_{loc}^{\infty}(\R^d)$ 
and that $\bodisp>0$ almost everywhere. Such $\bodisp$ will be called 
"admissible". Our main results, in addition, require that 
\begin{align*}
\begin{split}
   \assDisp \quad \bodisp(k)=\bodisp(-k)\qquad \text{and }\qquad \bodisp(k)=O(k^2)\ \ \text{as}\ \ \abs{k}\rightarrow\infty.
\end{split}
\end{align*} 
The main examples of dispersion relations satisfying $\assDisp$ are $\bodisp(k)\equiv 1$ and 
$\bodisp(k)=\sqrt{m^2+k^2}$, $m\geq 0$, which lead to the Fr{\"o}hlich Hamiltonian~\cite{GriesemerWuensch} 
and to the Nelson model~\cite{Nelson1964}, respectively. The Hamiltonian $H_0$ is positive and self-adjoint on 
$D(H_0)=D(p^2\otimes\ID)\cap D(\ID\otimes d\Gamma(\bodisp))$. Its form domain is 
given by the set $D(H_0^{1/2})=D(\abs{p}\otimes\ID)\cap D(\ID\otimes d\Gamma(\bodisp)^{1/2})$. 
The identity operator $\ID$ will be omitted from now on.

The interaction of the particle and the bosonic field is given in terms of annihilation and 
creation of bosons. The usual annihilation and creation operators in Fock space associated 
with some vector $f\in L^2(\R^d)$ will be denoted by $a(f)$ and $a^{*}(f)$. They are closed, 
adjoint to each other with $D(a(f)) = D(a^{*}(f)) \supset D(d\Gamma(\bodisp)^{1/2})$ if 
$f\bodisp^{-1/2}\in L^2(\R^d)$, Lemma~\ref{cor:SpecialEstAnnihilOp}, which means that $a(f)$ 
and $a^*(f)$ are well-defined on $D(H_0)$ and $D(H_0^{1/2})$ for any admissible 
$\bodisp$. These operators obey the canonical commutation relations 
$[a(f),a^{*}(g)] = \sprod{f}{g}$ (others vanish), which are operator equations on $D(d\Gamma(\bodisp))$ if $\bodisp^{1/2}f\in L^2(\R^d)$,
Lemma~\ref{lm:ladderMapsOnDomOfSquSecQuant}. The symmetric field operators
\begin{align*}
     \phi(f) := a(f) + a^{*}(f),\qquad \pi(f) := \phi(if)
\end{align*}
are essentially self-adjoint on $D(d\Gamma(\bodisp))$ by Nelson's commutator 
theorem and by Lemma~\ref{cor:SpecialEstAnnihilOp}, provided that $\bodisp^{1/2}f\in L^2(\R^d)$, 
and they satisfy the commutation relations 
\begin{align*}
     [\phi(f) ,\phi(g)] = 2i\Ima \sprod{f}{g},\qquad [\phi(f) ,\pi(g)] = 2i\Rea \sprod{f}{g},
\end{align*}
provided additionally that $f\bodisp^{-1/2}\in L^2(\R^d)$. The (self-adjoint) closures 
of the operators $\phi(f)$ and $\pi(f)$ will be denoted by the same symbols.

We will have occasion to work with generalized annihilation and creation operators $a(F)$ 
and $a^{*}(F)$ that are operators in $\hilbert$ rather than $\fock$. Here 
$F:L^2(\R^d,dx)\to L^2(\R^d,dx)\otimes L^2(\R^d,dk)$ is a linear operator. In the simplest 
case $F\varphi=\varphi\otimes f$ for some $f\in L^2(\R^d,dk)$ and then 
$a^{\#}(F)=\ID\otimes a^{\#}(f)$ is the usual annihilation or creation operator in $\fock$. 
Often, but not always, the operator $F$ will be defined in terms of some function 
$(x,k)\mapsto F_x(k)$, denoted by $F$ as well, through the equation 
$(F\varphi)(x,k)=\varphi(x)F_x(k)$. In this case $(a^{\#}(F)\Psi)(x)=a^{\#}(F_x)\Psi(x)$. 
Typically $F_x(k)=e^{-ikx}f(k)$ where $f \in L^2(\R^d)$ and then the operator norm of $F$ 
equals the norm of $f$ in $L^2(\R^d)$. See Appendix B of~\cite{GriesemerWuensch} for the 
definition of $a^{\#}(F)$ in the general case. For computations it is sometimes useful to 
expand $a(F)$ and $a^{*}(F)$ in terms of $a_k$ and $a_k^*$ by
\begin{align*}
   a(F)=\int \left(F(k)^{*}\otimes a_k\right)\, dk, 
   \quad a^*(F)=\int \left(F(k)\otimes a_k^{*}\right)\, dk,
\end{align*}
where $F(k)$ denotes a bounded operator on the particle space $L^2(\R^d,dx)$ and 
\begin{align*}
   (a_k\Psi)^{(n)}(x,k_1,...,k_n) &= (n+1)^{1/2}\ \Psi^{(n+1)}(x,k,k_1,...,k_n), \\
   (a_k^*\Psi)^{(n)}(x,k_1,...,k_n)&= n^{-1/2}\sum\limits_{j=1}^{n}\delta(k-k_j)\ \Psi^{(n-1)}(x,k_1,...,k_{j-1},k_{j+1},...,k_n).
\end{align*}
The canonical commutation relations then read $[a_k,a_{k'}^*]=\delta(k-k')$ and $[a_k,a_{k'}]=0$. In terms of $a_k$ and $a_k^{*}$ we have 
\begin{align*}
   d\Gamma(\bodisp) &= \int \bodisp(k)a_k^*a_k\, dk, \\
   \norm{d\Gamma(\bodisp)^{1/2}\Psi} &= \left(\int \bodisp(k)\norm{a_k\Psi}^2\, dk\right)^{1/2}.
\end{align*}

We are now ready to introduce the ultraviolet regularized Hamiltonian 
corresponding to \eqref{formal-Nelson}. For any $\uv<\infty$ we define 
$H_{\uv}:D(H_0)\subset\hilbert\rightarrow\hilbert$ by
\begin{align*}
      H_{\uv}:=H_0+\phi(G_{\uv}),
\end{align*}
where
\begin{align*}
   G_{\uv,x}(k):=\f(k)e^{-ikx}\chi_{\uv}(k).
\end{align*}
Here $\chi_{\uv}$ denotes the characteristic function of the set $\{k\in\R^d|\ \abs{k}\leq\uv\}$. 
On the form factor $\f:\R^d\mapsto\C$, we impose the conditions
\begin{align*}
   &\formFac \quad \f \in L^2_{loc}(\R^d) \quad \text{and} \quad  \f(k)=\f(-k), \\
   &\regThr \quad \int_{\R^d} \frac{\abs{\f(k)}^2}{(1+k^2)^{2}} dk< \infty, \\
   &\regTwo \quad \exists\ \irFix\geq 0 \ :\ \int_{|k|\leq\irFix}\frac{\abs{v(k)}^2}{\bodisp(k)}dk<\infty 
            \ \text{ and }\ m_{\irFix}:=\inf\limits_{|k|\geq\irFix}\bodisp(k)>0,
\end{align*}
which are satisfied for example by the Nelson model, with massless or massive bosons, in three space 
dimensions. In one and two space dimensions, the massive Nelson model satisfies 
the condition $\int \abs{\f(k)}^2 (1+k^2)^{-1}dk < \infty$, which is stronger then $\regThr$ and whose consequences are discussed in Section~\ref{sec:Nelson-low}. These models are therefore covered by the methods and results in  \cite{GriesemerWuensch}, see also \cite{DissWuensch2017}.
Since, by $\formFac$, $\regTwo$, and Lemma~\ref{cor:SpecialEstAnnihilOp}, the 
operator $\phi(G_{\uv})(d\Gamma(\bodisp)+1)^{-1/2}$ is bounded, it follows 
that $\phi(G_{\uv})$ is infinitesimally $H_0$-bounded and hence, by Kato-Rellich, 
$H_{\uv}$ is self-adjoint on $D(H_0)$ for all $\uv<\infty$. The symmetry condition in $\formFac$ simplifies some computations, but for the main results it is inessential.

Following Nelson we now introduce a two-parameter family of unitary 
transformations, called Gross-transformations, by
\begin{align*}
   U_{\ir,\uv}=e^{i\pi(B_{\ir,\uv})},\qquad \irFix\leq\ir<\uv\leq\infty,
\end{align*}
where
\begin{align*}
   B_{\ir,\uv,x}(k):=-\frac{\f(k)}{\bodisp(k)+k^2}e^{-ikx}\chi_{\uv}(k)\chibar_{\ir}(k)=-\frac{1}{\bodisp(k)+k^2}G_{\uv,x}(k)\chibar_{\ir}(k)
\end{align*}
and $\chibar_{\ir}:=1-\chi_{\ir}$.  We will use $kB_{\ir,\uv}$ and $\abs{k}^sB_{\ir,\uv}$ to 
denote the functions $kB_{\ir,\uv,x}(k)$ and $\abs{k}^sB_{\ir,\uv,x}(k)$, respectively. 
The lower cutoff $\ir$ will be chosen sufficiently large in Lemma~\ref{lm:KLMNCondition1} 
and Theorem~\ref{thm:ExistenceOfNRLimit}, below.
In all other results the size of $\ir$ is inessential and the only conditions 
$\irFix\leq\ir<\uv$ will not be repeated.
Note that the condition $B_{\ir,\infty} \in L^2(\R^d,dk)$ is equivalent to 
$\regThr$, which means that $\regThr$ cannot be weakened.
Note moreover that  
\begin{align*}
   U_{\ir,\uv}\rightarrow U_{\ir,\infty}\quad (\uv\to\infty)
\end{align*}
strongly in $\hilbert$, which follows from Lemma \ref{lm:strongWeyl}.

Assuming~$\assDisp$, $\formFac$, $\regThr$, and~$\regTwo$, we can use 
Lemma~\ref{lm:transformPhi}, Lemma~\ref{lm:transformDGamma}, 
Lemma~\ref{lm:transformPSquGen}, and the identity 
\begin{align*}
   p\cdot a^*(kB_{\ir,\uv}) + a(kB_{\ir,\uv})\cdot p  = a^*(kB_{\ir,\uv})\cdot p+ p\cdot a(kB_{\ir,\uv}) -\phi(k^2B_{\ir,\uv}),
\end{align*}
to find that 
\begin{align*}
   U_{\ir,\uv}p^2U_{\ir,\uv}^* &= p^2-2a^*(kB_{\ir,\uv})\cdot p-2p\cdot a(kB_{\ir,\uv})+2a^*(kB_{\ir,\uv})a(kB_{\ir,\uv}) \\
               &\phantom{=}+a(kB_{\ir,\uv})^2+a^*(kB_{\ir,\uv})^2+\phi(k^2B_{\ir,\uv})+\norm{kB_{\ir,\uv}}^2, \\
   U_{\ir,\uv}d\Gamma(\bodisp)U_{\ir,\uv}^* &= d\Gamma(\bodisp)+\phi(\bodisp B_{\ir,\uv})+\norm{\bodisp^{1/2}B_{\ir,\uv}}^2, \\
   U_{\ir,\uv}\phi(G_{\uv})U_{\ir,\uv}^* &= \phi(G_{\uv})+2\sprod{B_{\ir,\uv}}{G_{\uv}} 
\end{align*}
on $D(H_0)$ as long as $\uv<\infty$. In view of $(\bodisp(k)+k^2)B_{\ir,\uv}+G_{\uv}=G_{\ir}$ the field 
operators $\phi(\cdot)$ in the above equations add up to $\phi(G_{\ir})$. 
For the same reason the scalar terms add up to $E_{\ir}-E_{\uv}$ where
\begin{align*}
   E_{\uv}:=\int_{\abs{k}\leq\uv}\frac{\abs{\f(k)}^2}{\bodisp(k)+k^2}\, dk.
\end{align*}
Note that $E_{\uv}$ may diverge as $\uv\to\infty$, and it does diverge for 
the Nelson model. Note also that there is no divergence for $\abs{k}\rightarrow 0$ 
because of Assumption~$\regTwo$. We therefore define
\begin{align*}
   H_{\ir,\uv}':=\ &U_{\ir,\uv}H_{\uv}U_{\ir,\uv}^*+E_{\uv} \notag \\
             =\ &p^2+d\Gamma(\bodisp)+\phi(G_{\ir})-2a^*(kB_{\ir,\uv})\cdot p-2p\cdot a(kB_{\ir,\uv})\notag \\
               &+a(kB_{\ir,\uv})^2+a^*(kB_{\ir,\uv})^2+2a^*(kB_{\ir,\uv})a(kB_{\ir,\uv})+E_{\ir} 
\end{align*}
which is self-adjoint on $D(H_0)$ for $\uv<\infty$ under the assumptions made in 
this section. Let $V_{\ir,\uv}$ denote the interaction part of $H_{\ir,\uv}'$, 
so that $H_{\ir,\uv}'=H_0+V_{\ir,\uv}$.

%
%


\section{Construction of the norm-resolvent limit}
\label{sec:construct}

In this section we describe the construction of the Nelson Hamiltonian associated with~(1) 
in the generalized setup given by our Assumptions~$\assDisp$, $\formFac$, $\regThr$,  
and~$\regTwo$. Similar constructions given in the literature concern the special case of the Nelson Hamiltonian only \cite{Nelson1964,AmmariZied2000,HHS2005,GNV2006}. 
There is an abstract part in our argument that is summarized by Theorem~\ref{thm:AmariLikeTheorem} in Appendix~\ref{app:AmariLikeTheorem}. We apply this theorem to the quadratic form 
\begin{multline}
   \lefteqn{W_{\ir,\uv}(\Psi):=\sprod{\Psi}{V_{\ir,\uv}\Psi} =2\Rea\big\{\sprod{\Psi}{a(G_{\ir})\Psi}-2\sprod{a(kB_{\ir,\uv})\Psi}{p\Psi} }\\ 
 +\norm{a(kB_{\ir,\uv})\Psi}^2 + \sprod{a^{*}(kB_{\ir,\uv})\Psi}{a(kB_{\ir,\uv})\Psi}\big\} + E_{\ir}\norm{\Psi}^2, \label{equ:interactForm}
\end{multline}
defined on $D(H_0^{1/2})=D(\abs{p})\cap D(d\Gamma(\bodisp)^{1/2})$. The subsequent Lemmas~\ref{lm:KLMNCondition1} 
and~\ref{lm:KLMNCondition2} verify the hypotheses of Theorem~\ref{thm:AmariLikeTheorem}, and Theorem~\ref{thm:ExistenceOfNRLimit}, below, summarizes the result. Notice that much simpler and more direct characterizations of the Hamiltonian are possible under more restrictive assumptions on the decay of $\f$ \cite{GriesemerWuensch}. 

In the special case of the massive three-dimensional Nelson model the following two lemmas agree with Lemma 5 from~\cite{Nelson1964}.
  
\begin{lemma}
\label{lm:KLMNCondition1}
Assume $\assDisp$, $\formFac$, $\regThr$, $\regTwo$, and let 
$\norm{\f\chibar_{\irFix}(1+k^2)^{-1/2}\bodisp^{-1/4}}<\infty$. Then, for 
all $\eps>0$ there is a $\ir<\infty$ and a $b<\infty$, such that
\begin{align*}
   \abs{W_{\ir,\uv}(\Psi)}\leq \eps \norm{H_0^{1/2}\Psi}^2+b\norm{\Psi}^2
\end{align*}
for all $\Psi\in D(H_0^{1/2})$ and all $\uv<\infty$.
\end{lemma}

\begin{proof}
Note that $\norm{\f\chibar_{\irFix}(1+k^2)^{-1/2}\bodisp^{-1/4}}<\infty$ implies 
$\norm{kB_{\ir,\infty}\bodisp^{-1/4}}<\infty$, and $\norm{kB_{\ir,\infty}\bodisp^{-1/2}}<\infty$. 
We estimate the terms of~\eqref{equ:interactForm} one by one. For the first term, 
using Lemma~\ref{cor:SpecialEstAnnihilOp}, we find that for every $\eps>0$, 
$$
   \abs{\sprod{\Psi}{a(G_{\ir})\Psi}}\leq \norm{\Psi} \norm{\frac{G_{\ir}}{\sqrt{\bodisp}}} \norm{d\Gamma(\bodisp)^{1/2}\Psi}\leq \eps\norm{H_0^{1/2}\Psi}^2+\frac{1}{4\eps}\norm{\frac{G_{\ir}}{\sqrt{\bodisp}}}^2 \norm{\Psi}^2,
$$
where $\norm{G_{\ir}\bodisp^{-1/2}}<\infty$ by~$\formFac$ and~$\regTwo$.
Similarly, for the second and third terms, 
\begin{align*}
   \abs{\sprod{p\Psi}{a(kB_{\ir,\uv})\Psi}}\leq \norm{p\Psi} \norm{\frac{kB_{\ir,\uv}}{\sqrt{\bodisp}}} \norm{d\Gamma(\bodisp)^{1/2}\Psi}\leq \frac{1}{2}\norm{\frac{kB_{\ir,\infty}}{\sqrt{\bodisp}}}\ \norm{H_0^{1/2}\Psi}^2,
\end{align*}
and 
\begin{align*}
   \norm{a(kB_{\ir,\uv})\Psi}^2\leq \norm{\frac{kB_{\ir,\infty}}{\sqrt{\bodisp}}}^2\norm{H_0^{1/2}\Psi}^2.
\end{align*}
For the fourth term of~\eqref{equ:interactForm}, let $N_{\ir}:=d\Gamma(\chibar_{\ir})$ 
and note that for $\Psi\in D(H_0^{1/2})$ 
\begin{align}
\label{equ:estLargeMomNumberOp}
   N_{\ir}\leq\frac{1}{m_{\irFix}}d\Gamma(\bodisp).
\end{align}
Then, with the help of Lemma~\ref{lm:SpecialEstSquaredAnnihilOp} and Cauchy-Schwarz,
\begin{align*}
    \abs{\sprod{a^{*}(kB_{\ir,\uv})\Psi}{a(kB_{\ir,\uv})\Psi}} 
    &= \abs{\sprod{(1+N_{\ir})^{1/2}\Psi}{(1+N_{\ir})^{-1/2}a(kB_{\ir,\uv})^2\Psi}} \\
    &\leq\left(1+\frac{1}{m_{\irFix}}\right)^{1/2}\norm{kB_{\ir,\uv}\bodisp^{-1/4}}^2\norm{(1+d\Gamma(\bodisp))^{1/2}\Psi}^2 \\
    &\leq\left(1+\frac{1}{m_{\irFix}}\right)^{1/2}\norm{kB_{\ir,\infty}\bodisp^{-1/4}}^2\left(\norm{H_0^{1/2}\Psi}^2+\norm{\Psi}^2\right).
\end{align*}
Since $\norm{kB_{\ir,\infty}\bodisp^{-1/2}}$ and $\norm{kB_{\ir,\infty}\bodisp^{-1/4}}$ 
become arbitrarily small for $\ir$ 
sufficiently large, the lemma is proved. 
\end{proof}

\begin{lemma}
\label{lm:KLMNCondition2}
Assume $\assDisp$, $\formFac$, $\regThr$, $\regTwo$, and let 
$\norm{\f\chibar_{\irFix}(1+k^2)^{-1/2}\bodisp^{-1/4}}<\infty$. Then, for all $\Psi\in D(H_0^{1/2})$, 
we have
\begin{align*}
   \abs{W_{\ir,\uv_1}(\Psi)-W_{\ir,\uv_2}(\Psi)}\leq C_{\ir,\uv_1,\uv_2} \norm{(H_0+1)^{1/2}\Psi}^2,
\end{align*}
where $C_{\ir,\uv_1,\uv_2}\rightarrow 0$ as $\uv_1,\uv_2\rightarrow\infty$.
\end{lemma}
\begin{proof}
Note that $\norm{\f\chibar_{\irFix}(1+k^2)^{-1/2}\bodisp^{-1/4}}<\infty$ implies 
$\norm{kB_{\ir,\infty}\bodisp^{-1/4}}<\infty$, and 
$\norm{kB_{\ir,\infty}\bodisp^{-1/2}}<\infty$. By definition of $W_{\ir,\uv}$ and
with $N_{\ir}:=d\Gamma(\chibar_{\ir})$,
\begin{align*}
   &W_{\ir,\uv_1}(\Psi)-W_{\ir,\uv_2}(\Psi) \notag \\
\begin{split}
   &=2\Rea \Big\{-2\sprod{p\Psi}{a(kB_{\uv_2,\uv_1})\Psi}+\norm{a(kB_{\ir,\uv_1})\Psi}^2-\norm{a(kB_{\ir,\uv_2})\Psi}^2 \\
   &\phantom{=====}+\sprod{(1+N_{\ir})^{1/2}\Psi}{(1+N_{\ir})^{-1/2}\left(a(kB_{\ir,\uv_1})^2-a(kB_{\ir,\uv_2})^2\right)\Psi}\Big\}.
\end{split}
\end{align*}
We estimate term by term with the help of the arguments from the proof of 
Lemma~\ref{lm:KLMNCondition1}. Using Lemma~\ref{cor:SpecialEstAnnihilOp},
\begin{align*}
   \abs{\sprod{p\Psi}{a(kB_{\uv_2,\uv_1})\Psi}}\leq \frac{1}{2}\norm{\frac{kB_{\uv_2,\uv_1}}{\sqrt{\bodisp}}}\ \norm{H_0^{1/2}\Psi}^2,
\end{align*}
and likewise,
\begin{align*}
   &\abs{\norm{a(kB_{\ir,\uv_1})\Psi}^2-\norm{a(kB_{\ir,\uv_2})\Psi}^2} \notag \\
   &=\abs{\norm{a(kB_{\ir,\uv_1})\Psi}-\norm{a(kB_{\ir,\uv_2})\Psi}}\left(\norm{a(kB_{\ir,\uv_1})\Psi}+\norm{a(kB_{\ir,\uv_2})\Psi}\right) \notag \\
   &\leq\norm{a(kB_{\uv_2,\uv_1})\Psi}\left(\norm{a(kB_{\ir,\uv_1})\Psi}+\norm{a(kB_{\ir,\uv_2})\Psi}\right) \notag \\
   &\leq 2 \norm{\frac{kB_{\uv_2,\uv_1}}{\sqrt{\bodisp}}} \norm{\frac{kB_{\ir,\infty}}{\sqrt{\bodisp}}}\ \norm{H_0^{1/2}\Psi}^2,
\end{align*}
where $\norm{kB_{\uv_2,\uv_1}\bodisp^{-1/2}}\rightarrow 0$ as 
$\uv_1,\uv_2\rightarrow\infty$. Finally, using 
Equation~\eqref{equ:estLargeMomNumberOp} and 
Lemma~\ref{lm:SpecialEstSquaredAnnihilOp},
\begin{align*}
   &\abs{\sprod{(1+N_{\ir})^{1/2}\Psi}{(1+N_{\ir})^{-1/2}\left(a(kB_{\ir,\uv_1})^2-a(kB_{\ir,\uv_2})^2\right)\Psi}} \notag \\
   &\leq\norm{(1+N_{\ir})^{1/2}\Psi}\norm{(1+N_{\ir})^{-1/2}a(kB_{\uv_2,\uv_1})a(kB_{\ir,\uv_1}+kB_{\ir,\uv_2})\Psi} \notag \\
   &\leq\left(1+\frac{1}{m_{\irFix}}\right)^{1/2}\norm{\frac{kB_{\uv_2,\uv_1}}{\sqrt[4]{\bodisp}}}
    \norm{\frac{k(B_{\ir,\uv_1}+B_{\ir,\uv_2})}{\sqrt[4]{\bodisp}}}\norm{(1+d\Gamma(\bodisp))^{1/2}\Psi} \notag \\
   &\leq 2\left(1+\frac{1}{m_{\irFix}}\right)^{1/2}
    \norm{\frac{kB_{\uv_2,\uv_1}}{\sqrt[4]{\bodisp}}}   \norm{\frac{kB_{\ir,\infty}}{\sqrt[4]{\bodisp}}} \norm{(H_0+1)^{1/2}\Psi}^2,
\end{align*}
where $\norm{kB_{\uv_2,\uv_1}\bodisp^{-1/4}}\rightarrow 0$ as $\uv_1,\uv_2\rightarrow\infty$.
\end{proof}

By the two lemmas above,  $W_{\ir,\uv}$ satisfies the hypotheses of 
Theorem~\ref{thm:AmariLikeTheorem} for $\ir$ large enough. By Lemma~\ref{lm:strongWeyl}, the assumption of part (iii) of Theorem~\ref{thm:AmariLikeTheorem} is satisfied as well. We therefore conclude:

\begin{theorem}
\label{thm:ExistenceOfNRLimit}
Assume $\assDisp$, $\formFac$, $\regThr$, $\regTwo$, and let 
$\norm{\f\chibar_{\irFix}(1+k^2)^{-1/2}\bodisp^{-1/4}}<\infty$. Then, for $\ir$ large enough, 
there exists a unique, self-adjoint, semi-bounded operator $H_{\ir,\infty}'$ 
with $D(H_{\ir,\infty}')\subset D(H_0^{1/2})$ associated to the quadratic form
\begin{align*}
   \sprod{H_0^{1/2}\Psi}{H_0^{1/2}\Psi}+\limop{\uv}W_{\ir,\uv}(\Psi),
\end{align*}
which is defined for all $\Psi\in D(H_0^{1/2})$. Moreover, we have
\begin{align*}
   H_{\ir,\uv}' &\rightarrow H_{\ir,\infty}' \qquad (\uv\rightarrow\infty), \\
   H_{\uv}+E_{\uv} &\rightarrow H:=U_{\ir,\infty}^*H_{\ir,\infty}'U_{\ir,\infty} \qquad (\uv\rightarrow\infty)
\end{align*}
in the norm-resolvent sense and $D(H)\subset U_{\ir,\infty}^*D(H_0^{1/2})$.
\end{theorem}


\section{Regularity of domain vectors}   
\label{sec:domain}

From Theorem~\ref{thm:ExistenceOfNRLimit}, we know that $U_{\ir,\infty}^*D(H_0^{1/2})$ 
is the form domain of the operator $H$ and hence $D(H)\subset U_{\ir,\infty}^*D(H_0^{1/2})$.
This section is devoted to the study of the regularity of elements 
$\Psi\in U_{\ir,\infty}^*D(H_0^{1/2})$ based on the decay of $\f$, or more precisely, 
based on the decay of $f=(\bodisp+k^2)^{-1}\f$. 
To stress this point and to simplify the notation we formulate most results of this 
section as mapping properties of some unitary operator
$$
         U_{f}^{*} = e^{-i\pi(F)},\qquad F_x(k) = f(k) e^{-ik x}
$$   
under suitable assumptions on the decay of $f\in L^2_{\rm loc}(\R^d)$ at $|k|=\infty$. Let 
$$
         U_{f,\uv}^{*} = e^{-i\pi(F_\uv)},\qquad F_{\uv,x} = F_x\chi_\uv.
$$   
Unless otherwise stated, $\bodisp$ is assumed to be "admissible", which
means that $\bodisp\in L_{loc}^{\infty}(\R^d)$ and $\bodisp(k)>0$ almost everywhere. 
The main results of this section are Theorems~\ref{thm:domainSubset3D4D} 
and~\ref{thm:domainIntersec3D4D}. 

\begin{theorem}
\label{thm:domainSubset3D4D}
Assume that $f\in L^2(\R^d)$ is a function with 
$|k|^sf,|k|^sf\bodisp^{-1/2}\in L^2(\R^d)$ for some 
$s\in [0,1]$. Then 
\begin{align*}
   U_{f}^*D(H_0^{1/2})\subset D(\abs{p}^{s}).
\end{align*}
\end{theorem}

\begin{proof}
As a preparation, we record the following three facts:
\begin{enumerate}
\item From Lemma~\ref{lm:SpecEstGenAnnihilOp}, we know that
\begin{align*}
   \norm{[a^{\#}(F),\abs{p}^s](d\Gamma(\bodisp)+1)^{-1/2}}
   &=\norm{a^{\#}([F,\abs{p}^s])(d\Gamma(\bodisp)+1)^{-1/2}} \\
   &\leq\max\left\{\norm{\abs{k}^sf},\norm{\abs{k}^sf\bodisp^{-1/2}}\right\}.
\end{align*}
\item $D(H_0^{1/2})\cap\hilbert_0$ is a core of $H_0^{1/2}$: If $\Psi\in D(H_0^{1/2})$ then 
$\Psi_n:=\chi(N\leq n)\Psi\in D(H_0^{1/2})\cap\hilbert_0$ and it is 
straightforward to show that $\norm{H_0^{1/2}(\Psi-\Psi_n)}\rightarrow 0$ 
as $n\rightarrow\infty$.
\item If $\Psi\in D(H_0^{1/2})\cap\hilbert_0$ then 
$e^{-i\pi(F)t}\Psi\in D(\abs{p}^s)$ and 
$t\mapsto\abs{p}^se^{-i\pi(F)t}\Psi$ is real-analytic: 

It is well-known that for $\Psi\in\hilbert_0$
\begin{align*}
   e^{-i\pi(F)t}\Psi=\sum\limits_{n\geq 0}\frac{(-it)^n}{n!}\pi(F)^n\Psi,
\end{align*}
where this series is absolutely convergent for all $t\in\R$. The proof 
of this fact combined with Fact 1 shows that
\begin{align}
\label{equ:domainSubsetExpansion}
   \sum\limits_{n\geq 0}\frac{(-it)^n}{n!}\abs{p}^s\pi(F)^n\Psi=\sum\limits_{n\geq 0}\frac{(-it)^n}{n!}\pi(F)^n\abs{p}^s\Psi+\sum\limits_{n\geq 0}\frac{(-it)^n}{n!}[\abs{p}^s,\pi(F)^n]\Psi  
\end{align}
is absolutely convergent for all $t\in\R$. By the closedness of the operator $\abs{p}^s$, 
it follows that $e^{-i\pi(F)t}\Psi\in D(\abs{p}^s)$ and that 
$\abs{p}^se^{-i\pi(F)t}\Psi$ is given 
by~\eqref{equ:domainSubsetExpansion}.
\end{enumerate}
We are now ready to prove the theorem. For all $\Phi\in D(\abs{p})$ and 
all $\Psi\in D(H_0^{1/2})\cap\hilbert_0$ it follows from Facts 1 and 3 that
\begin{align*}
   &\sprod{\Phi}{\abs{p}^se^{-i\pi(F)}\Psi-e^{-i\pi(F)}\abs{p}^s\Psi}=\sprod{\Phi}{e^{-i\pi(F)(1-t)}\abs{p}^se^{-i\pi(F)t}\Psi} \Big|_{t=0}^{t=1} \notag \\
   &=\int_{0}^{1}\sprod{\Phi}{e^{-i\pi(F)(1-t)}[i\pi(F),\abs{p}^s]e^{-i\pi(F)t}\Psi} dt.
\end{align*}
We write this in the form
\begin{align}
\label{equ:proofAntilinFunct}
   \sprod{\abs{p}^s\Phi}{e^{-i\pi(F)}\Psi}=\sprod{\Phi}{e^{-i\pi(F)}\abs{p}^s\Psi}+\int_{0}^{1}\sprod{\Phi}{e^{-i\pi(F)(1-t)}[i\pi(F),\abs{p}^s]e^{-i\pi(F)t}\Psi} dt.
\end{align}
For given $\Phi\in D(\abs{p})$ this equation extends to all 
$\Psi\in D(H_0^{1/2})$. This follows from Fact 2 and the boundedness of 
$\abs{p}^s(H_0+1)^{-1/2}$ and 
$[i\pi(F),\abs{p}^s]e^{-i\pi(F)t}(H_0+1)^{-1/2}$ which 
follows from Fact 1 and Lemma~\ref{lm:strongWeyl}. 
Equation~\eqref{equ:proofAntilinFunct} now shows that for all 
$\Psi\in D(H_0^{1/2})$
\begin{align*}
   \Phi\mapsto\sprod{\abs{p}^s\Phi}{e^{-i\pi(F)}\Psi}
\end{align*}
is a bounded anti-linear functional defined on $D(\abs{p})$, which is 
a core of $\abs{p}^s$. Since the operator $\abs{p}^s$ is self-adjoint, 
it follows that $e^{-i\pi(F)}\Psi\in D(\abs{p}^s)$.
\end{proof}

\begin{lemma}
\label{commute-p-Weyl}
If $f\in L^2_{loc}(\R^d)$, then for all  $s\in [0,1]$ and all finite $\Lambda>0$ 
\begin{align*}
\left\|(1+d\Gamma(\bodisp))^{-1/2} \left[|p|^s,  U_{f,\uv}^{*} \right](1+d\Gamma(\bodisp))^{-1/2}\right\|
 \leq 2 \||k|^{s}f\chi_\uv\bodisp^{-1/2}\|\cdot (1+\|\bodisp^{1/2}f\chi_\uv\|).
\end{align*}
\end{lemma}

\noindent
\emph{Remark:} With only one of the factors $(1+d\Gamma(\bodisp))^{-1/2}$ on the right hand side of the inequality,  a similar bound involving the norm $\|\abs{k}^sf\chi_\uv\|$ is true. Such a bound is not good enough for our purpose.

\begin{proof}
Let $F=F_\uv$ for notational simplicity. On $D(|p|^s)$ in the sense of quadratic forms we have
\begin{align*}
   &|p|^s e^{-i\pi(F)} - e^{-i\pi(F)} |p|^s = \int_0^1 e^{-i\pi(F)(1-t)}[i\pi(F),|p|^s] e^{-i\pi(F)t}\, dt\\
   &= \int_0^1 e^{-i\pi(F)(1-t)}a([|p|^s,F]) e^{-i\pi(F)t}\, dt - \int_0^1 e^{-i\pi(F)(1-t)}a^{*}([F,|p|^s]) e^{-i\pi(F)t}\, dt.
\end{align*}
Using now Corollary~\ref{cor:transformSquRootDGammaSandwich} and 
Lemma~\ref{lm:SpecEstGenAnnihilOp} the desired estimate easily follows.
\end{proof}

Lemma~\ref{lm:domainIntersec3D4DLimitExists} 
and Corollary~\ref{lm:domainIntersec3D4DHelpLemma}, below, 
are the main tools for the proof of Theorem~\ref{thm:domainIntersec3D4D}.

\begin{lemma}
\label{lm:domainIntersec3D4DLimitExists}
Assume $f\in L^2_{loc}(\R^d)$ and in addition that 
$\bodisp^{1/2}f$, $\abs{k}^sf\bodisp^{-1/2}\in L^2(\R^d)$ for some $s\in [0,1]$.
Then, the following limit exists: 
\begin{align*}
   \limop{\uv}\left(1+d\Gamma(\bodisp)\right)^{-1/2}\abs{p}^sU_{f,\uv}^* (1+H_0)^{-1/2}.
\end{align*}
\end{lemma}
\begin{proof}
Let $R_{\bodisp} := (1+d\Gamma(\bodisp))^{-1/2}$ and for $0\leq\uv_1,\uv_2 <\infty$ let
\begin{align*}
   U_{f,\uv_1,\uv_2}^*:=e^{-i\pi(F_{\uv_2}-F_{\uv_1})}=e^{-i\pi(F\chi_{\uv_2}(1-\chi_{\uv_1}))}.
\end{align*} 
From Lemma~\ref{lm:strongWeyl}, from the assumptions on $f$, 
and from the boundedness of $|p|^s(1+H_0)^{-1/2}$ it is clear that the limit 
$\limop{\uv} R_{\bodisp} U_{f,\uv}^* \abs{p}^s(H_0+1)^{-1/2}$
exists. It remains to show that $R_{\bodisp} [\abs{p}^s, U_{f,\uv}^* ] (H_0+1)^{-1/2}$ has a 
limit as $\uv \to \infty$. In fact, we will show that 
\begin{equation}
\label{limit-to-prove}
      \lim_{\uv\to\infty} R_{\bodisp} \left[\abs{p}^s, U_{f,\uv}^* \right] R_{\bodisp}
\end{equation}
exists. To this end we use that $[\pi(F_{\uv_1}),\pi(F_{\uv_2})]=0$ and hence, 
\begin{equation}
\label{BCH}
       U_{f,\uv_2}^* -  U_{f,\uv_1}^* = ( U_{f,\uv_2}^* U_{f,\uv_1} -1) U_{f,\uv_1}^* =  ( U_{f,\uv_1,\uv_2}^* -1) U_{f,\uv_1}^*.
\end{equation}
From \eqref{BCH}, it follows that
\begin{align}
   &R_{\bodisp} \left[\abs{p}^s, U_{f,\uv_2}^{*} \right] R_{\bodisp} - R_{\bodisp} \left[\abs{p}^s, U_{f,\uv_1}^* \right] R_{\bodisp} \notag \\
   &= R_{\bodisp} \left[\abs{p}^s, U_{f,\uv_1,\uv_2}^* \right]  U_{f,\uv_1}^*  R_{\bodisp} + R_{\bodisp} ( U_{f,\uv_1,\uv_2}^* -1) \left[\abs{p}^s, U_{f,\uv_1}^{*} \right] R_{\bodisp} \notag \\
   &= R_{\bodisp} \left[\abs{p}^s, U_{f,\uv_1,\uv_2}^* \right]  R_{\bodisp} \cdot R_{\bodisp}^{-1} U_{f,\uv_1}^*  R_{\bodisp} \notag \\
   \label{equ:limExistsCalc} 
   &\phantom{=====}+ R_{\bodisp} ( U_{f,\uv_1,\uv_2}^* -1)R_{\bodisp}^{-1} \cdot R_{\bodisp} \left[\abs{p}^s, U_{f,\uv_1}^{*} \right] R_{\bodisp} .
\end{align}
where the identities $R_{\bodisp} \cdot R_{\bodisp}^{-1} =1 = R_{\bodisp}^{-1} \cdot R_{\bodisp}$ have been inserted in the last equation.
By Lemma \ref{commute-p-Weyl} and by Corollary~\ref{cor:transformSquRootDGammaSandwich},
\begin{eqnarray}
      R_{\bodisp} \left[\abs{p}^s, U_{f,\uv_1,\uv_2}^* \right]  R_{\bodisp}    &\to & 0 \\
      R_{\bodisp} ( U_{f,\uv_1,\uv_2}^* -1)R_{\bodisp}^{-1} &\to & 0
\end{eqnarray} 
as $\uv_1,\uv_2\to \infty$, while 
\begin{eqnarray}
   \sup_{\uv} \| R_{\bodisp}^{-1} U_{f,\uv}^*  R_{\bodisp}  \| &<&\infty\\ 
   \label{equ:limExistsUniBound} \sup_{\uv} \| R_{\bodisp} \left[\abs{p}^s, U_{f,\uv}^{*} \right] R_{\bodisp} \| &<&\infty. 
\end{eqnarray}
From~\eqref{equ:limExistsCalc} to~\eqref{equ:limExistsUniBound}, 
it follows that the limit \eqref{limit-to-prove} exists.
\end{proof}
\begin{corollary}
\label{lm:domainIntersec3D4DHelpLemma} 
Suppose the hypotheses of Lemma~\ref{lm:domainIntersec3D4DLimitExists} are
satisfied for some $s\in [0,1]$, and moreover, that $U_f^*\Psi\in D(|p|^s)$ 
for some $\Psi\in D(H_0^{1/2})$. Then, 
\begin{align*}
   \left(1+d\Gamma(\bodisp)\right)^{-1/2}\abs{p}^sU_{f,\uv}^*\Psi \rightarrow \left(1+d\Gamma(\bodisp)\right)^{-1/2}\abs{p}^sU_{f}^*\Psi \qquad (\uv\rightarrow\infty).
\end{align*}
\end{corollary}
\begin{proof}
Let $R_{\bodisp}:=\left(1+d\Gamma(\bodisp)\right)^{-1/2}$ and note that $R_{\bodisp}\abs{p}^s\subset\abs{p}^s R_{\bodisp}$.
We know that $U_{f,\uv}^*\Psi\rightarrow U_{f}^*\Psi$ as $\uv\rightarrow\infty$
and that $U_{f,\uv}^*\Psi\in D(H_0^{1/2})\subset D(\abs{p}^s)$ by Lemma~\ref{lm:transformPGen}.
It follows that 
\begin{align*}
   R_{\bodisp}U_{f,\uv}^*\Psi\rightarrow R_{\bodisp}U_{f}^*\Psi\qquad (\uv\to\infty),
\end{align*}
where $\limop{\uv}\abs{p}^sR_{\bodisp}U_{f,\uv}^*\Psi$ exists,
by Lemma~\ref{lm:domainIntersec3D4DLimitExists}. Since $\abs{p}^s$ is a
closed operator, we conclude that $R_{\bodisp}U_{f}^*\Psi \in D(|p|^s)$ and that 
\begin{align*}
   \limop{\uv}\abs{p}^sR_{\bodisp}U_{f,\uv}^*\Psi=\abs{p}^sR_{\bodisp}U_{f}^*\Psi. 
\end{align*}
The corollary now follows from $R_{\bodisp}\abs{p}^s\subset\abs{p}^s R_{\bodisp}$ and from
the assumptions that $U_{f}^*\Psi\in D(|p|^s)$.
\end{proof}
%
\begin{theorem}
\label{thm:domainIntersec3D4D}
Assume $f$, $\bodisp^{1/2}f$, $f\bodisp^{-1/2}\in L^2(\R^d)$ and in addition that, for some $s\in [0,1]$, 
$\abs{k}^{s/2}f$, $\abs{k}^sf\bodisp^{-1/2}\in L^2(\R^d)$,
while $\abs{k}^sf\not\in L^2(\R^d)$. Then
\begin{align*}
   U_{f}^*D(H_0^{1/2})\cap D(\abs{p}^s)=\{ 0\}.
\end{align*}
\end{theorem}
\begin{proof}
Let $\Psi\in D(H_0^{1/2})$ and suppose that $U_{f}^*\Psi\in D(\abs{p}^s)$ for some $s\in [0,1]$ 
for which the hypotheses of the theorem are satisfied. Then, of course, 
$\norm{\abs{p}^sU_{f}^*\Psi}<\infty$. We will show that this is not true unless $\Psi=0$. 
Our proof is based on the identity  
\begin{align*}
   \norm{\abs{p}^sU_{f}^*\Psi}=\limop[0]{\eps}\limop{\uv}\norm{(1+\eps d\Gamma(\bodisp))^{-1/2}\abs{p}^sU_{f,\uv}^*\Psi},
\end{align*}
which follows from Corollary~\ref{lm:domainIntersec3D4DHelpLemma} and monotone convergence. Concerning the necessity of the regularization in terms of the operator $(1+\eps d\Gamma(\bodisp))^{-1/2}$ see the remark following Lemma~\ref{commute-p-Weyl} and recall that $\abs{k}^sf\not\in L^2(\R^d)$ by assumption.
We are now going to write the vector $(1+\eps d\Gamma(\bodisp))^{-1/2}\abs{p}^sU_{f,\uv}^*\Psi$ 
as sum of terms, where all but one have a norm that is uniformly bounded in $\eps$ and $\uv$. 
We call them "good terms". The norm of the remaining term diverges unless $\Psi =0$. This 
will complete the proof.

Using
\begin{align}   \label{equ:domainIntersec3D4DCommutator}
   \abs{p}^sU_{f,\uv}^*=U_{f,\uv}^*\abs{p}^s+\int_{0}^{1} dt\ U_{f,\uv}^*(1-t)\ [i\pi(F_{\uv}),\abs{p}^s]\ U_{f,\uv}^*(t),
\end{align}
where $U_{f,\uv}^*(t):=e^{-i\pi(F_{\uv})t}$ (note that $U_{f,\uv}^*(1)=U_{f,\uv}^*$ and $i\pi(F_{\uv})=a(F_{\uv})-a^*(F_{\uv})$), 
we obtain three summands:
\begin{align}
   &(1+\eps d\Gamma(\bodisp))^{-1/2}\abs{p}^sU_{f,\uv}^*\Psi \notag \\
   &=(1+\eps d\Gamma(\bodisp))^{-1/2}U_{f,\uv}^*\abs{p}^s\Psi+(1+\eps d\Gamma(\bodisp))^{-1/2}\int_{0}^{1} dt\ U_{f,\uv}^*(1-t)a([F_{\uv},\abs{p}^s])U_{f,\uv}^*(t)\Psi \notag \\
   \label{equ:domainIntersec3D4DToDiverge1} &\phantom{====}-(1+\eps d\Gamma(\bodisp))^{-1/2}\int_{0}^{1} dt\ U_{f,\uv}^*(1-t)a^*([F_{\uv},\abs{p}^s])U_{f,\uv}^*(t)\Psi.
\end{align}
The first two terms are good terms. Indeed, the norm of the first summand can 
directly be estimated by $\norm{\abs{p}^s\Psi}$, which is finite. Using  
Lemmas~\ref{lm:SpecEstGenAnnihilOp} and~\ref{lm:transformSquRootDGamma} the norm 
of the second term can be estimated as follows
\begin{align*}
   &\Big\|(1+\eps d\Gamma(\bodisp))^{-1/2}\int_{0}^{1}dt\ U_{f,\uv}^*(1-t)a([F_{\uv},\abs{p}^s])U_{f,\uv}^*(t)\Psi\Big\| \notag \\
   &\leq\int_{0}^{1} dt\ \norm{a([F_{\uv},\abs{p}^s])U_{f,\uv}^*(t)\Psi } \leq\int_{0}^{1} dt\ \norm{\frac{\abs{k}^sF_{\uv,x}}{\sqrt{\bodisp}}}\ \norm{d\Gamma(\bodisp)^{1/2}U_{f,\uv}^*(t)\Psi} \notag \\
   &\leq\norm{\frac{\abs{k}^sf}{\sqrt{\bodisp}}}\left(1+\norm{\bodisp^{1/2}f}\right)\ \norm{(1+d\Gamma(\bodisp))^{1/2}\Psi},
\end{align*}
where the supremum with respect to $t\in [0,1]$ and $\uv>0$ was taken 
in the last step. 

It remains to show the divergence of the norm of the third term 
of~\eqref{equ:domainIntersec3D4DToDiverge1}. Using that $e^{ikx}\abs{p}^se^{-ikx}=\abs{p-k}^s$, this term reads
\begin{align}
\label{equ:domainIntersec3D4DToDiverge1b} 
   (1+\eps d\Gamma(\bodisp))^{-\frac{1}{2}}\int_{0}^{1} dt\ \int dk\ F_{\uv,x}(k)U_{f,\uv}^*(1-t)a_k^*A_{p,k}U_{f,\uv}^*(t)\Psi,
\end{align}
where we defined
\begin{align*}
   A_{p,k}:=\abs{p-k}^s-\abs{p}^s
\end{align*}
for short. Recall, that 
\begin{align}
\label{equ:domainIntersec3D4DPEst}
   \abs{A_{p,k}}=\abs{\abs{p-k}^s-\abs{p}^s}\leq\abs{k}^s
\end{align}
for $p,k\in\R^d$ and $s\in [0,1]$, and that
\begin{align}
\label{equ:domainIntersec3D4DCommRel}
   U_{f,\uv}^*(1-t)a_k^*U_{f,\uv}(1-t)=a_k^*-(1-t)\overline{F_{\uv,x}(k)}.
\end{align}
Identity~\eqref{equ:domainIntersec3D4DCommRel} is used now to commute $U_{f,\uv}^*(1-t)$ and $a_k^*$. 
Then, \eqref{equ:domainIntersec3D4DToDiverge1b} becomes
\begin{align}
   &(1+\eps d\Gamma(\bodisp))^{-\frac{1}{2}}\int_{0}^{1} dt \int dk\ F_{\uv,x}(k)a_k^*U_{f,\uv}^*(1-t)A_{p,k}U_{f,\uv}^*(t)\Psi \notag \\
   \label{equ:domainIntersec3D4DToDiverge2} &\phantom{} -(1+\eps d\Gamma(\bodisp))^{-\frac{1}{2}}\int_{0}^{1} dt \int dk (1-t)\abs{F_{\uv}(k)}^2U_{f,\uv}^*(1-t)A_{p,k}U_{f,\uv}^*(t)\Psi,
\end{align}
where we use $\abs{F_{\uv}(k)}:=\abs{F_{\uv,x}(k)}$ to point out, that this value does not depend on $x$ anymore. 
We will do the same for norms containing an $F$. 
The second term of~\eqref{equ:domainIntersec3D4DToDiverge2} is another good term, due 
to~\eqref{equ:domainIntersec3D4DPEst}, its norm is bounded by
\begin{align*}
   &\int_{0}^{1} dt\ (1-t) \int dk\ \abs{F_{\uv}(k)}^2 \norm{A_{p,k}U_{f,\uv}^*(t)\Psi} \notag \\
   &\leq \int_{0}^{1} dt\ (1-t) \int dk\ \abs{k}^{s}\abs{F_{\uv}(k)}^2 \norm{U_{f,\uv}^*(t)\Psi}\leq \norm{\abs{k}^{s/2}f}^2 \norm{\Psi},
\end{align*}
which is finite by the assumptions of the theorem. 

We continue to analyze the first term of~\eqref{equ:domainIntersec3D4DToDiverge2}
and show its divergence. For short, define 
$\eta_{\uv,p,k}(t):=U_{f,\uv}^*(1-t)A_{p,k}U_{f,\uv}^*(t)\Psi$. We consider 
the squared norm of this first term of~\eqref{equ:domainIntersec3D4DToDiverge2}, write
it as an inner product, commute the ladder operators, and use the pull-through formulas for 
$d\Gamma(\bodisp)$. This calculation reads
\begin{align}
   &\Big\|(1+\eps d\Gamma(\bodisp))^{-1/2}\int_{0}^{1} dt\ \int dk\ F_{\uv,x}(k)a_k^*\eta_{\uv,p,k}(t)\Big\|^2 \notag \\
   &=\int_{0}^{1} dt \int_{0}^{1} dt' \int dk \int dk'\ \overline{F_{\uv,x}(k)} F_{\uv,x}(k') \sprod{a_k^*\eta_{\uv,p,k}(t)}{(1+\eps d\Gamma(\bodisp))^{-1}a_{k'}^*\eta_{\uv,p,k'}(t')} \notag \\
   &=\int_{0}^{1} dt \int_{0}^{1} dt' \int dk \int dk'\ \overline{F_{\uv,x}(k)} F_{\uv,x}(k') \notag \\
   &\phantom{=====}\sprod{a_{k'}\eta_{\uv,p,k}(t)}{(1+\eps (d\Gamma(\bodisp)+\bodisp(k)+\bodisp(k')))^{-1}a_{k}\eta_{\uv,p,k'}(t')} \notag \\
   \label{equ:domainIntersec3D4DToDiverge3} &\phantom{=}+\int_{0}^{1} dt \int_{0}^{1} dt' \int dk\ \abs{F_{\uv}(k)}^2\sprod{\eta_{\uv,p,k}(t)}{(1+\eps (d\Gamma(\bodisp)+\bodisp(k)))^{-1}\eta_{\uv,p,k}(t')} \\
   &=: \alpha + \beta, \notag
\end{align}
where $\alpha$ and $\beta$ denotes the two resulting summands. 

The term $\alpha$ is a good term as we now show: by Cauchy-Schwarz
\begin{align}
\label{equ:domainIntersec3D4DToConverge} 
   \abs{\alpha} \leq\int_{0}^{1} dt \int_{0}^{1} dt' \int dk \int dk'\ \abs{F_{\uv}(k)}\ \abs{F_{\uv}(k')}\ \norm{a_{k'}\eta_{\uv,p,k}(t)}\ \norm{a_{k}\eta_{\uv,p,k'}(t')}.
\end{align}
By definition of $\eta_{\uv,p,k}(t)$, by~\eqref{equ:domainIntersec3D4DPEst} 
and by the identity
\begin{align*}
   U_{f,\uv}(t)a_kU_{f,\uv}^*(t)=a_k+t\cdot F_{\uv,x}(k),
\end{align*}
used twice, we obtain
\begin{align}
   \norm{a_{k'}\eta_{\uv,p,k}(t)}&\leq\abs{k}^s\norm{a_{k'}U_{f,\uv}^*(t)\Psi}+\abs{k}^s\abs{F_{\uv}(k')}\ \norm{\Psi} \notag \\
   \label{equ:domainIntersec3D4DToConvergeSideCalc} &\leq\abs{k}^s\norm{a_{k'}\Psi}+\abs{k}^s\abs{F_{\uv}(k')}\ \norm{\Psi}.
\end{align}
We now multiply and divide $\norm{a_{k'}\Psi}$ by $\sqrt{\bodisp(k')}$,
analogously for the momentum $k$,
and insert~\eqref{equ:domainIntersec3D4DToConvergeSideCalc} 
into~\eqref{equ:domainIntersec3D4DToConverge}. Using Cauchy-Schwarz
we obtain
\begin{align*}
   \abs{\alpha} &\leq\norm{\frac{\abs{k}^sF_{\uv}}{\sqrt{\bodisp}}}^2\norm{d\Gamma(\bodisp)^{1/2}\Psi}^2+\norm{\abs{k}^{s/2}F_{\uv}}^4\norm{\Psi}^2 \notag \\
                &\phantom{=====}+2\norm{\frac{\abs{k}^sF_{\uv}}{\sqrt{\bodisp}}}\ \norm{\abs{k}^{s/2}F_{\uv}}^2\ \norm{d\Gamma(\bodisp)^{1/2}\Psi}\ \norm{\Psi} .
\end{align*}
This remains finite in the limit $\uv\rightarrow\infty$ by the 
assumptions of the theorem. Thus, $\alpha$ from 
Equation~\eqref{equ:domainIntersec3D4DToDiverge3} is a good term and the 
divergence has to be in the second term of~\eqref{equ:domainIntersec3D4DToDiverge3},
we called it $\beta$. 

It remains to show that $\beta$ diverges. To this end, let
\begin{align*}
   R_{\eps}(k):=\left(1+\eps\left(d\Gamma(\bodisp)+\bodisp(k)\right)\right)^{-1/2}
\end{align*}
for short. Then,
\begin{align*}
   \beta &=\int dk \abs{F_{\uv}(k)}^2\Big\|R_{\eps}(k)\int_{0}^{1} dt\ U_{f,\uv}^*(1-t)A_{p,k}U_{f,\uv}^*(t)\Psi\Big\|^2 \notag \\
         &=\int dk \abs{F_{\uv}(k)}^2\bigg\|R_{\eps}(k)A_{p,k}U_{f,\uv}^*\Psi+R_{\eps}(k)\int_{0}^{1} dt\ \bigg[U_{f,\uv}^*(1-t),A_{p,k}\bigg]U_{f,\uv}^*(t)\Psi\bigg\|^2. 
\end{align*}
Note, that $U_{f,\uv}^*(1-t)U_{f,\uv}^*(t)=U_{f,\uv}^*$ and, therefore, the first summand in the norm does 
not depend on $t$ anymore. If we define
\begin{align}
\label{equ:domainIntersec3D4DToDiverge4Beta1} \beta_1 &:=\int dk\ \abs{F_{\uv}(k)}^2\norm{R_{\eps}(k)A_{p,k}U_{f,\uv}^*\Psi}^2, \\
\label{equ:domainIntersec3D4DToDiverge4Beta2} \beta_2 &:=\int dk\ \abs{F_{\uv}(k)}^2\bigg\|R_{\eps}(k)\int_{0}^{1} dt\ \bigg[U_{f,\uv}^*(1-t),A_{p,k}\bigg]U_{f,\uv}^*(t)\Psi\bigg\|^2,
\end{align}
then $\beta\geq\frac{1}{2}\beta_1-2\beta_2$. To complete the proof, we
show that $\beta_2$ is a good term, while $\beta_1$ diverges.

For $\beta_2$, using Lemma~\ref{lm:SpecImpEstForDGammaGenLadderOp}, 
Corollary~\ref{cor:transformSquRootDGammaSandwich}, and a representation 
of the commutator analogously to~\eqref{equ:domainIntersec3D4DCommutator}, we find
\begin{align}
   \beta_2 &=\int d\ell \abs{F_{\uv}(\ell)}^2\bigg\|R_{\eps}(\ell)\int_{0}^{1} dt \int_{0}^{1} dr 
             (1-t)U_{f,\uv}^*(1-t)\bigg[A_{p,\ell},i\pi(F_{\uv})\bigg]U_{f,\uv}^*(1-r+tr)\Psi\bigg\|^2 \notag \\
   &\leq\int d\ell \abs{F_{\uv}(\ell)}^2\left(\int_{0}^{1} dt \int_{0}^{1} dr (1-t)
             \norm{\pi(\left[A_{p,\ell},F_{\uv}\right])U_{f,\uv}^*(1-r+tr)\Psi}\right)^2 \notag \\
   &\leq\int d\ell \abs{F_{\uv}(\ell)}^2\Big(\int_{0}^{1} dt \int_{0}^{1} dr (1-t)\cdot
             4\abs{\ell}^{s/2}\max\left\{\norm{\abs{k}^{s/2}f\chi_{\uv}},\norm{\abs{k}^{s/2}f\bodisp^{-1/2}\chi_{\uv}}\right\} \notag \\
   &\phantom{\int d\ell \abs{F_{\uv}(\ell)}^2.......}\left(1+\norm{\bodisp^{1/2}f\chi_{\uv}}(1-r+tr)\right)
             \norm{(1+d\Gamma(\bodisp))^{1/2}\Psi}\Big)^2 \notag \\
   &\leq 8\int d\ell \abs{\ell}^s \abs{F_{\uv}(\ell)}^2
         \max\left\{\norm{\abs{k}^{s/2}f\chi_{\uv}},\norm{\abs{k}^{s/2}f\bodisp^{-1/2}\chi_{\uv}}\right\} \notag \\
   &\phantom{\int d\ell\ \abs{F_{\uv}(\ell)}^2.......}\left(1+\norm{\bodisp^{1/2}f\chi_{\uv}}\right)
             \norm{(1+d\Gamma(\bodisp))^{1/2}\Psi}^2 \notag \\
   &\leq 8\norm{\abs{k}^{\frac{s}{2}}f}^2
         \max\left\{\norm{\abs{k}^{s/2}f},\norm{\abs{k}^{s/2}f\bodisp^{-1/2}}\right\}
         \left(1+\norm{\bodisp^{1/2}f}\right)^2\norm{(1+d\Gamma(\bodisp))^{\frac{1}{2}}\Psi}^2. \notag
\end{align}
This is finite by the hypotheses on $f$ and by
Assumption~$\assDisp$. In the last calculation, we used the
integration variable $\ell$ instead of $k$ for avoiding 
misunderstandings in view of the application of
Lemma~\ref{lm:SpecImpEstForDGammaGenLadderOp}.

It remains to analyze $\beta_1$ from Equation~\eqref{equ:domainIntersec3D4DToDiverge4Beta1}.
Using 
\begin{align*}
   (\abs{p-k}^s-\abs{p}^s)^2\geq (\abs{k}^{2s}-4\abs{p}^{s}\abs{k}^{s}) \qquad (s\in [0,1]),
\end{align*}
we obtain
\begin{align}
   \beta_1 &=\int dk \abs{F_{\uv}(k)}^2\sprod{R_{\eps}(k)U_{f,\uv}^*\Psi}{A_{p,k}^2R_{\eps}(k)U_{f,\uv}^*\Psi} \notag \\
   &\geq\int dk \abs{k}^{2s}\abs{F_{\uv}(k)}^2\norm{R_{\eps}(k)U_{f,\uv}^*\Psi}^2 
   \label{equ:domainIntersec3D4DToDiverge5} -4\int dk \abs{k}^{s}\abs{F_{\uv}(k)}^2\norm{R_{\eps}(k)\abs{p}^{\frac{s}{2}}U_{f,\uv}^*\Psi}^2.
\end{align}
The second term of~\eqref{equ:domainIntersec3D4DToDiverge5} is a good term, because,
by Fatou, Corollary~\ref{lm:domainIntersec3D4DHelpLemma}, and monotone convergence
\begin{align*}
   &\limsupop[0]{\eps}\ \limsupop{\uv}\int dk\ \abs{k}^{s}\abs{F_{\uv}(k)}^2\norm{R_{\eps}(k)\abs{p}^{s/2}U_{f,\uv}^*\Psi}^2 \notag \\
   &\leq\limsupop[0]{\eps}\int dk\ \abs{k}^{s}\abs{f(k)}^2\norm{R_{\eps}(k)\abs{p}^{s/2}U_{f}^*\Psi}^2 \notag \\
   &\leq\norm{\abs{k}^{s/2}f}^2\norm{\abs{p}^{s/2}U_{f}^*\Psi}^2,
\end{align*}
which is finite by our assumptions.

Finally, we look at the first term of~\eqref{equ:domainIntersec3D4DToDiverge5}.
For any $\uv_0\leq\uv$, by Fatou, Corollary~\ref{lm:domainIntersec3D4DHelpLemma}, 
and monotone convergence
\begin{align*}
   &\liminfop[0]{\eps}\ \liminfop{\uv}\int dk \abs{k}^{2s}\abs{F_{\uv}(k)}^2\norm{R_{\eps}(k)U_{f,\uv}^*\Psi}^2  \notag \\
   &\geq\liminfop[0]{\eps}\ \liminfop{\uv}\int_{\abs{k}\leq\uv_0} dk\ \abs{k}^{2s}\abs{f(k)}^2\norm{R_{\eps}(k)U_{f,\uv}^*\Psi}^2 \notag \\
   &\geq\liminfop[0]{\eps}\int_{\abs{k}\leq\uv_0} dk\ \abs{k}^{2s}\abs{f(k)}^2\norm{R_{\eps}(k)U_{f}^*\Psi}^2 \notag \\
   &=\norm{\abs{k}^sf\chi_{\uv_0}}^2\norm{\Psi}^2.
\end{align*}
Since $\norm{\abs{k}^sf\chi_{\uv_0}}$ diverges as $\uv_0\rightarrow\infty$,
we conclude that $\norm{\Psi}=0$ and the proof is complete.
\end{proof}
%

\section{The Nelson model in dimensions $d\leq 3$}
\label{sec:Nelson-low}

We now turn to the Nelson model in space dimensions $d\leq 3$, where 
\begin{align*}
   \f(k)=\frac{1}{\sqrt{\bodisp(k)}}\quad\text{with}\quad\bodisp(k)=\sqrt{m^2+k^2}.
\end{align*}
If $d=3$, then the Hypothesis $\assDisp$ and $\formFac - \regTwo$ are satisfied for any $m\geq 0$.
By Theorem~\ref{thm:ExistenceOfNRLimit}, the norm-resolvent limit $H$ of $H_{\uv}+E_{\uv}$ exists and its form domain is 
$U_{\ir,\infty}^*D(H_0^{1/2})$ for $\ir$ sufficiently large. By
Lemma~\ref{lm:transformSquRootDGamma} we know that 
$U_{\ir,\infty}^*D(d\Gamma(\bodisp)^{1/2})=D(d\Gamma(\bodisp)^{1/2})$ 
and hence $U_{\ir,\infty}^*D(H_0^{1/2})\subset D(d\Gamma(\bodisp)^{1/2})$. On the other hand,
Theorem~\ref{thm:domainSubset3D4D} shows that $U_{\ir,\infty}^*D(H_0^{1/2})\subset D(\abs{p}^s)$ for any $s<1$. This proves \eqref{equ:introResultDim3Inclus} in the introduction. Equation \eqref{equ:introResultDim3Intersec} follows from Theorem~\ref{thm:domainIntersec3D4D} with $s=1$.

If $d\in \{1,2\}$, then we need $m>0$ for $ \regTwo$ to hold. In fact, if $m=0$ then even the regularized Hamiltonian $H_\Lambda$ is unbounded below.
As far as $\regThr$ is concerned, we now have the stronger property
$$
      \int_{\R^d}\frac{|\f(k)|^2}{1+k^2}\, dk < \infty.
$$
It follows that $|k|B_x(k)$ is square-integrable for $d\leq 2$, which, by making a mild additional decay assumption allows us to prove that 
$D(H) = U_{\infty}^{*} D(H_0)$ along the line of arguments given in~\cite{GriesemerWuensch} 
The critical exponent $(5-d)/2$ in following theorem can now 
be understood on the basis of the heuristics argument given in the introduction for the regularity 
preserved by the Gross transform $U_{\infty}^{*} $. A detailed proof is given in \cite{DissWuensch2017}.
\begin{theorem}
For the massive Nelson Hamiltonian $H$ in dimensions $d\leq 2$ we have 
$$
   D(H) \subset \bigcap\limits_{0\leq s < (5-d)/2} D(\abs{p}^s)\cap D(d\Gamma(\bodisp))
$$
and 
$$
   D(H) \cap D(\abs{p}^{(5-d)/2})=\{0\}.
$$
\end{theorem}
%

\appendix

\section{Quadratic forms and resolvent convergence}   
\label{app:AmariLikeTheorem}

The following theorem contains the abstract part behind Theorem~\ref{thm:ExistenceOfNRLimit}. It agrees with Theorem~A.1 of~\cite{GriesemerWuensch} with the exception of additional statement that is needed here. There are similar theorems by Nelson~\cite{Nelson1964} and Ammari~\cite{AmmariZied2000}.

\begin{theorem}
\label{thm:AmariLikeTheorem}
Let $H_0\geq 0$ be a self-adjoint operator in the Hilbert space $\hilbert$ and 
let $\norm{\Psi}_0:=\norm{(H_0+1)^{1/2}\Psi}$ for $\Psi\in D(H_0^{1/2})$. For each 
$\uv<\infty$ let $W_{\uv}$ be a quadratic form defined on $D(H_0^{1/2})$ such that
\begin{itemize}
\item[(a)] for all $\Psi\in D(H_0^{1/2})$ and all $\uv<\infty$,
\begin{align*}
   \abs{W_{\uv}(\Psi)}\leq a\norm{\Psi}_0^2+b_{\uv}\norm{\Psi}^2,
\end{align*}
where $a<1$,
\item[(b)] for all $\Psi\in D(H_0^{1/2})$,
\begin{align*}
   \abs{W_{\uv}(\Psi)-W_{\uv'}(\Psi)}\leq C_{\uv,\uv'}\norm{\Psi}_0^2,
\end{align*}
where $C_{\uv,\uv'}\rightarrow 0$ as $\uv,\uv'\rightarrow\infty$.
\end{itemize}
Let $W_{\infty}(\Phi,\Psi):=\limop{\uv}W_{\uv}(\Phi,\Psi)$. Then, the following statements hold true:
\begin{itemize}
\item[(i)] Statement (a) extends to $\uv=\infty$ with some finite $b_{\infty}$, and for each $\uv\leq\infty$, there 
exists a self-adjoint, semibounded operator $H_{\uv}$ with $D(H_{\uv})\subset D(H_0^{1/2})$ and 
\begin{align*}
   \sprod{\Phi}{H_{\uv}\Psi}=\sprod{H_0^{1/2}\Phi}{H_0^{1/2}\Psi}+W_{\uv}(\Phi,\Psi)
\end{align*}
for all $\Phi\in D(H_0^{1/2})$ and $\Psi\in D(H_{\uv})$. 
\item[(ii)] For all $z\in\C\backslash\R$, 
\begin{align*}
   (H_{\uv}-z)^{-1}\rightarrow (H_{\infty}-z)^{-1} \qquad (\uv\rightarrow\infty)
\end{align*}
in the operator norm. 
\item[(iii)] If $U_{\uv}$, $0<\uv\leq \infty$, is a one-parameter family of unitary operator with $(H_0+1)^{-1/2}(U_{\uv} - U_{\infty})\to 0$ as $\uv\to\infty$, then for all $z\in\C\backslash\R$,
\begin{align*}
   (U_{\uv}^*H_{\uv}U_{\uv}-z)^{-1} \rightarrow (U_{\infty}^*H_{\infty}U_{\infty}-z)^{-1} \qquad (\uv\rightarrow\infty)
\end{align*}
in the operator norm.
\end{itemize}
\end{theorem}
\begin{proof}
We only prove part (iii). Parts (i) and (ii) are proven in \cite{GriesemerWuensch}. From the first part of the theorem (and its proof) it follows that 
\begin{align*}
   H_0\leq\frac{2}{1-a}(H_{\uv}+M), \quad \uv_0\leq\uv\leq\infty,
\end{align*}
with constants $a\in (0,1)$ and $M>0$. This implies that $(H_0+1)^{1/2}(H_\uv-z)^{-1}$, for $z\in \C\backslash\R$, is bounded uniformly in $\uv\geq \uv_0$. Let $R_{\uv}(z):=(H_\uv-z)^{-1}$. Then, for $z\in \C\backslash\R$, 
\begin{eqnarray*}
   \lefteqn{(U^{*}_{\uv} H_{\uv}U_{\uv} - z)^{-1} -  (U^{*}_{\infty} H_{\infty}U_{\infty} - z)^{-1} }\\ &=&  U^{*}_{\uv} R_{\uv}(z)U_{\uv} - U^{*}_{\infty} R_{\infty}(z)U_{\infty} \\
 &=&   U^{*}_{\uv} R_{\uv}(z)(U_{\uv} - U_{\infty}) +  U^{*}_{\uv} (R_{\uv}(z) -  R_{\infty}(z))U_{\infty} + (U^{*}_{\uv} - U^{*}_{\infty}) R_{\infty}(z)U_{\infty}
\end{eqnarray*}
where all three terms vanish in the limit $\uv\to \infty$ by the assumption on $U_\uv$ combined with the boundedness of  $(H_0+1)^{1/2}R_\uv(z)$, and by the convergence $R_{\uv}(z)\to R_{\infty}(z)$ known from part (ii).
\end{proof}

\section{Creation and annihilation operators}   
\label{app:ladderOperators}

This section contains bounds on creation and annihilation operators 
relative to other, unbounded operators. Prototypical for many of the 
following bounds are
\begin{equation*}
      \|a(f) \Psi\| \leq \|f\|\cdot \|N^{1/2}\Psi\|,\qquad \|a^{*}(f) \Psi\| \leq \|f\|\cdot \|(N+1)^{1/2}\Psi\|,
\end{equation*} 
which follow immediately from the definitions of $a(f)$ and $a^{*}(f)$ 
and from the canonical commutation relations, see, e.g. \cite{GriesemerWuensch}.
As in main part of this paper, we will always assume that $\bodisp\in L_{loc}^{\infty}(\R^d)$ 
and that $\bodisp>0$ almost everywhere.

\begin{lemma}
\label{cor:SpecialEstAnnihilOp}
Let $\alpha\geq\frac{1}{2}$ and $f\in L^2(\R^d)$. 
Then, for all $\Psi\in D(d\Gamma(\bodisp)^{\alpha})$,
\begin{align*} 
   \norm{a(f)\Psi}  &\leq\norm{f\bodisp^{-\alpha}}\ \norm{d\Gamma(\bodisp)^{\alpha}\Psi}, \\
   \norm{a^*(f)\Psi}&\leq\max\left\{\norm{f},\norm{f\bodisp^{-\alpha}}\right\}\norm{(1+d\Gamma(\bodisp))^{\alpha}\Psi}.
\end{align*}
\end{lemma}
\begin{proof}
From $a(f)\Psi = \int \overline{f(k)}a_k\Psi\, dk$ and 
Cauchy-Schwarz it follows that 
\begin{align*}
      \|a(f)\Psi\| = \left\|\int \overline{f(k)}\bodisp(k)^{-\alpha}\bodisp(k)^{\alpha}a_k\Psi\, dk\right\|
                       \leq \| f\bodisp^{-\alpha}\| \sprod{\Psi}{d\Gamma(\bodisp^{2\alpha})\Psi}.
\end{align*}
The condition $2\alpha\geq 1$ implies that 
$d\Gamma(\bodisp^{2\alpha})\leq d\Gamma(\bodisp)^{2\alpha}$ and thus, 
the first inequality of the lemma is proved. The second one 
follows from the first and from
$\norm{a^*(f)\Psi}^2=\norm{a(f)\Psi}^2+\norm{f}^2\norm{\Psi}^2$.
\end{proof}

\begin{lemma}
\label{lm:ladderMapsOnDomOfSquSecQuant}
Suppose that $f, f\bodisp^{-1/2}\in L^2(\R^d)$. 
Then, for all $\Psi\in D(d\Gamma(\bodisp))$,
\begin{align*} 
   \norm{d\Gamma(\bodisp)^{1/2}a(f)\Psi} &\leq\norm{f\bodisp^{-1/2}}\ \norm{d\Gamma(\bodisp)\Psi}, \\
   \norm{d\Gamma(\bodisp)^{1/2}a^*(f)\Psi}&\leq \left(\norm{\bodisp^{1/2}f}+\norm{f\bodisp^{-1/2}}\right)
                                                     \norm{(1+d\Gamma(\bodisp))\Psi}.
\end{align*}
\end{lemma}
\begin{proof}
Let $\Psi\in D(d\Gamma(\bodisp))$ and note that  $\norm{d\Gamma(\bodisp)^{1/2}a_k\Psi} \leq \norm{(d\Gamma(\bodisp)+\omega(k))^{1/2}a_k\Psi} = \norm{a_kd\Gamma(\bodisp)^{1/2}\Psi}$ because  $\omega\geq 0$. Using this and Cauchy-Schwarz, we obtain
\begin{align*} 
   \norm{d\Gamma(\bodisp)^{1/2}a(f)\Psi} &\leq\int dk\ \abs{f(k)}\norm{d\Gamma(\bodisp)^{1/2}a_k\Psi} \\
   &\leq\left(\int dk\ \frac{\abs{f(k)}^2}{\bodisp(k)}\right)^{1/2} \left(\int dk\ \bodisp(k)\norm{a_kd\Gamma(\bodisp)^{1/2}\Psi}\right)^{1/2} \\
   &=\norm{\frac{f}{\sqrt{\bodisp}}}\cdot \norm{d\Gamma(\bodisp)\Psi},
\end{align*}
which proves the first estimate. To prove the second one, we set $R_{\bodisp}:=(1+d\Gamma(\bodisp))^{-1}$ and we use that 
\begin{align*} 
   d\Gamma(\bodisp)^{1/2}a^*(f)R_{\bodisp}=d\Gamma(\bodisp)^{1/2}R_{\bodisp}a^*(f)-d\Gamma(\bodisp)^{1/2}R_{\bodisp}a^*(\bodisp f)R_{\bodisp},
\end{align*}
because $[a^*(f),d\Gamma(\bodisp)]=-a^*(\bodisp f)$. Since $\|d\Gamma(\omega)^{1/2}R_\omega^{1/2}\| \leq 1$, 
Lemma~\ref{cor:SpecialEstAnnihilOp} implies
\begin{align*} 
   \norm{d\Gamma(\bodisp)^{1/2}a^*(f)R_{\bodisp}}\leq \norm{R_{\bodisp}^{1/2}a^*(f)}+\norm{R_{\bodisp}^{1/2}a^*(\bodisp f)}\leq \norm{f/\sqrt{\bodisp}}+\norm{\sqrt{\bodisp}f},
\end{align*}
which completes the proof.
\end{proof}

\begin{lemma}
\label{lm:SpecialEstSquaredAnnihilOp}
Let  $\chi_{\Omega}$ be the characteristic function of some measurable $\Omega\subset\R^d$ and let $N_{\Omega}:=d\Gamma(\chi_{\Omega})$. Then, for any $f,g\in L^2(\R^d)$,
\begin{align*} 
   \norm{(1+N_{\Omega})^{-1/2}a(f\chi_{\Omega})a(g\chi_{\Omega})(1+d\Gamma(\bodisp))^{-1/2}}
   \leq \norm{f\chi_{\Omega}\bodisp^{-1/4}} \norm{g\chi_{\Omega}\bodisp^{-1/4}}.
\end{align*}
\end{lemma}

\begin{proof}
Let $\fock(\Omega)$ and $\fock(\Omega^c)$ denote the symmetric Fock spaces over $L^2(\Omega)$ and $L^2(\Omega^c)$, respectively, and let $\eta$ denote the (normalized) vacuum in any of these Fock spaces. Let $U$ denote the unitary transformation $U:\fock(\R^d)\rightarrow\fock(\Omega)\otimes\fock(\Omega^c)$ defined by  $U\eta = \eta\otimes \eta$ and 
$$
    Ua(h)U^* = a(h\chi_{\Omega})\otimes\ID+\ID\otimes a(h\chi_{\Omega^c}).
$$
Then
\begin{align*} 
   Ud\Gamma(\bodisp)U^* &= d\Gamma(\bodisp\chi_{\Omega})\otimes\ID+\ID\otimes d\Gamma(\bodisp\chi_{\Omega^c}), \\
   UN_{\Omega}U^* &= N\otimes\ID,
\end{align*}
where $N$ is the number operator on $\fock(\Omega)$ and $h\in L^2(\R^d)$. Conjugating by $U$ the operator whose norm is to be estimated, we find that  
\begin{align} 
   &\norm{(1+N_{\Omega})^{-1/2}a(f\chi_{\Omega})a(g\chi_{\Omega})(1+d\Gamma(\bodisp))^{-1/2}}\nonumber \\
   &=\Big\|\left((1+N)^{-1/2}\otimes\ID\right)\big(a(f\chi_{\Omega})\otimes\ID\big)\big(a(g\chi_{\Omega})\otimes\ID\big)\nonumber \\
   &\phantom{======}\Big(\ID\otimes\ID+d\Gamma(\bodisp\chi_{\Omega})\otimes\ID+\ID\otimes d\Gamma(\bodisp\chi_{\Omega^c})\Big)^{-1/2}\Big\|\nonumber\\
&\leq\norm{(1+N)^{-1/2}a(f\chi_{\Omega})a(g\chi_{\Omega})(1+d\Gamma(\bodisp\chi_{\Omega}))^{-1/2}}_{\fock(\Omega)},\label{norm-to-bound}
\end{align}
where the index of the norm indicates the space on which the operator acts. Following Nelson in the proof of Lemma 5 of~\cite{Nelson1964}, we obtain for $\Psi\in D(d\Gamma(\bodisp)^{1/2})\cap\fock(\Omega)$
\begin{align*} 
   \norm{(1+N)^{-1/2}a(f\chi_{\Omega})a(g\chi_{\Omega})\Psi}_{\fock(\Omega)}
   \leq\norm{f\chi_{\Omega}\bodisp^{-1/4}} \norm{g\chi_{\Omega}\bodisp^{-1/4}}\norm{d\Gamma(\bodisp)^{1/2}\Psi}_{\fock(\Omega)},
\end{align*}
which, combined with \eqref{norm-to-bound}, proves the lemma.
%
\end{proof}

\begin{lemma}
\label{lm:SpecEstGenAnnihilOpForBoundedOp}
Let $p:=-i\nabla_x$, and let $B:L^2(\R^d,dx)\rightarrow L^2(\R^d,dx)\otimes L^2(\R^d,dk)$ denote an operator of the form $\varphi(x)\mapsto b(p,k)e^{-ikx}\varphi(x)$ with some measurable function $(p,k)\mapsto b(p,k)$. Then 
\begin{align*} 
   \norm{a(B)\Psi}   &\leq C_1(b) \norm{d\Gamma(\bodisp)^{1/2}\Psi}, \\
   \norm{a^*(B)\Psi} &\leq\max\left\{C_1(b),C_2(b)\right\} \norm{(1+d\Gamma(\bodisp))^{1/2}\Psi},
\end{align*}
where
\begin{align*} 
   C_1(b) := \left(\sup\limits_{p\in\R^d}\int dk\ \frac{\abs{b(p-k,k)}^2}{\bodisp(k)}\right)^{1/2}, 
   \quad C_2(b) := \left(\sup\limits_{p\in\R^d}\int dk\ \abs{b(p-k,k)}^2\right)^{1/2}.
\end{align*}
\end{lemma}
\begin{proof}
As in the proof of Lemma~\ref{cor:SpecialEstAnnihilOp}, we see that
$\norm{a(B)\Psi} \leq \norm{B\bodisp^{-1/2}}\ \norm{d\Gamma(\bodisp)^{1/2}\Psi}$
but now $\|B\bodisp^{-1/2}\|$ denotes the norm of the operator $B\bodisp^{-1/2}$ . Using $(B\bodisp^{-1/2}\varphi)(p,k) = b(p,k)\hat{\varphi}(p+k)$ and making  the substitution $p+k\to k$ one easily finds that $\|B\bodisp^{-1/2}\| \leq C_1(b)$ which proves the first estimate. 

The proof of the second estimate is more involved: By normal ordering 
\begin{align}
    \norm{a^*(B)\Psi}^2 &=\int dk\int dk' \sprod{B(k)a_k^{*}\Psi}{B(k')a_{k'}^{*}\Psi}\nonumber \\
    &=\int dk \|B(k)\Psi\|^2 + \int dk\int dk'\  \sprod{B(k)a_{k'}\Psi}{B(k')a_{k}\Psi}, \label{normal-order}
\end{align} 
where
\begin{align*}
      \int dk \|B(k)\Psi\|^2 &= \int |b(p,k)|^2 \|\Psi(p+k)\|^2\, dk dp\\
     &=  \int |b(p-k,k)|^2 \|\Psi(p)\|^2\, dk dp\\
     &\leq \left(\sup\limits_{p\in\R^d}\int dk\ |b(p-k,k)|^2 \right) \|\Psi\|^2.
\end{align*} 
Let $e_k$ denote the operator of multiplication with $e^{-ikx}$ in the particle position space. Then $B(k)=b(p,k)e_k = e_k b(p-k,k)$. For the integrand of the second term of~\eqref{normal-order}, we therefore have
\begin{align*} 
   &\sprod{B(k)a_{k'}\Psi}{B(k')a_{k}\Psi} = \sprod{ b(p,k)e_k a_{k'}\Psi}{b(p,k')e_{k'} a_{k}\Psi} \\
    &= \sprod{b(p-k-k',k)e_{k'}^{*} a_{k'}\Psi}{b(p-k-k',k')e_{k}^{*} a_{k}\Psi}\\
    &=  \int dp \sprod{b(p-k-k',k) a_{k'}\Psi(p-k')}{b(p-k-k',k') a_{k}\Psi(p-k)}\\
    &\leq \int dp |b(p-k-k',k)| \|a_{k'}\Psi(p-k')\| \cdot  |b(p-k-k',k')| \|a_{k}\Psi(p-k)\|\\
    &= \int dp \frac{|b(p-k-k',k)|}{\sqrt{\bodisp(k)}} \sqrt{\bodisp(k')}\|a_{k'}\Psi(p-k')\| \frac{ |b(p-k-k',k')|}{\sqrt{\bodisp(k')}} \sqrt{\bodisp(k)} \|a_{k}\Psi(p-k)\|
\end{align*}
The integrand here is of the form $f(p,k,k')f(p,k',k)$. Using that $f(p,k,k')f(p,k',k)\leq \frac{1}{2}(f(p,k,k')^2+ f(p,k',k)^2)$ and integrating with respect to $k$ and $k'$, we arrive at the bound
\begin{align*} 
      &\int dk\int dk' \int dp \frac{ |b(p-k-k',k')|}{\bodisp(k')} \bodisp(k) \|a_{k}\Psi(p-k)\|^2\\
     &\quad\leq \sup\limits_{p\in\R^d}\int dk\ \frac{\abs{b(p-k,k)}^2}{\bodisp(k)}  \|d\Gamma(\bodisp)^{1/2}\Psi\|^2
\end{align*}
for the second term of \eqref{normal-order}. This completes the proof.
\end{proof}


\begin{lemma}
\label{lm:SpecEstGenAnnihilOp}
Let $p:=-i\nabla_x$ and $F_x(k):=f(k)e^{-ikx}$ with 
$\abs{k}^sf\in L^2(\R^d)$ for some $s\in [0,1]$. Then, for all 
$\Psi\in D(d\Gamma(\bodisp)^{1/2})$,
\begin{align*} 
   \norm{a([\abs{p}^s,F])\Psi}   &\leq\norm{\abs{k}^s f\bodisp^{-1/2}}\ \norm{d\Gamma(\bodisp)^{1/2}\Psi}, \\
   \norm{a^*([\abs{p}^s,F])\Psi} &\leq\max\left\{\norm{\abs{k}^sf},\norm{\abs{k}^s f\bodisp^{-1/2}}\right\}\ \norm{(1+d\Gamma(\bodisp))^{1/2}\Psi}.
\end{align*}
\end{lemma}
%
%
\begin{proof}
Note that 
\begin{align*} 
   [\abs{p}^s,F_x(k)]=\abs{p}^sf(k)e^{-ikx}-f(k)e^{-ikx}\abs{p}^s=\left(\abs{p}^s-\abs{p+k}^s\right)f(k)e^{-ikx}
\end{align*}
and $\abs{\ \abs{p\pm k}^s-\abs{p}^s}\leq\abs{k}^s$
for all $p,k\in\R^d$ and $s\in [0,1]$. The Lemma now follows  
by Lemma~\ref{lm:SpecEstGenAnnihilOpForBoundedOp}.
\end{proof}

\begin{lemma}
\label{lm:SpecImpEstForDGammaGenLadderOp}
Let $p:=-i\nabla_x$, $\ell\in\R^d$, and $F_x(k):=f(k)e^{-ikx}$ 
with $\abs{k}^{s/2}f\in L^2(\R^d)$ for some $s\in [0,1]$. Then, for all 
$\Psi\in D(d\Gamma(\bodisp)^{1/2})$,
\begin{align*} 
   \norm{a([\abs{p\pm\ell}^s-\abs{p}^s,F])\Psi} &\leq 2\abs{\ell}^{s/2}\norm{\abs{k}^{s/2}f\bodisp^{-1/2}}\ \norm{d\Gamma(\bodisp)^{1/2}\Psi}, \\
   \norm{a^*([\abs{p\pm\ell}^s-\abs{p}^s,F])\Psi} &\leq 2\abs{\ell}^{s/2}\max\left\{\norm{\abs{k}^{s/2}f},\norm{\abs{k}^{s/2}f\bodisp^{-1/2}}\right\} \norm{\left(d\Gamma(\bodisp)+1\right)^{1/2}\Psi}. 
\end{align*}
\end{lemma}
%
%
\begin{proof}
We have 
\begin{align*} 
   [\abs{p\pm\ell}^s-\abs{p}^s,F_x(k)]=D_s(p,k,\ell)f(k)e^{-ikx}
\end{align*}
with
\begin{align*} 
   D_s(p,k,\ell):=\abs{p\pm\ell}^s-\abs{p}^s-\abs{p+k\pm\ell}^s+\abs{p+k}^s.
\end{align*}
For $p,k,\ell\in\R^d$ and $s\in [0,1]$,
\begin{align*} 
   \abs{D_s(p,k,\ell)}\leq 2\abs{k}^s \quad \text{and}\quad \abs{D_s(p,k,\ell)}\leq 2\abs{\ell}^s.
\end{align*}
That implies that
\begin{align*} 
\label{equ:EstForMoreAbsValues}
   \abs{D_s(p,k,\ell)}\leq 2\min\{\abs{k}^s,\abs{\ell}^s\} = 2\min\{\abs{k}^{s/2},\abs{\ell}^{s/2}\}^2  
                      \leq 2\abs{k}^{s/2}\abs{\ell}^{s/2}.
\end{align*}
This upper bound is independent of the first argument $p$.
The lemma now follows from Lemma~\ref{lm:SpecEstGenAnnihilOpForBoundedOp}.
\end{proof}

\section{Mapping properties of Weyl operators}     
\label{app:mappingProps}

This appendix collects some important identities and estimates in connection with the Weyl operator 
$e^{i\pi(f)}$ and the dressing transform of Gross.
The first lemma recalls a well-known fact, see e.g. Proposition 5.2.4.(1) in Bratteli and 
Robinson~\cite{BratteliRobinsonII}. The second lemma generalizes Lemma~B.4 of~\cite{GriesemerWuensch}. 
As in main part of this paper, we will always assume an admissible $\bodisp$, which means 
that $\bodisp\in L_{loc}^{\infty}(\R^d)$ and that $\bodisp>0$ almost everywhere.

\begin{lemma} 
\label{lm:transformPhi} 
Let $f,g\in L^2(\R^d)$. Then $e^{i\pi(f)}D(\phi(g))=D(\phi(g))$ and
\begin{align*}  
  e^{i\pi(f)} \phi(g)e^{-i\pi(f)}=\phi(g)+2\Rea\sprod{g}{f} \qquad \textrm{on } D(\phi(g)).
\end{align*}
\end{lemma}

\begin{lemma}
\label{lm:strongWeyl}
Let $f,g\in L^2(\R^d)$. Then
\begin{align*}
   \norm{(e^{i\pi(f)}-e^{i\pi(g)})(1+d\Gamma(\bodisp)))^{-1/2}}\leq 2\max\left\{\norm{f-g},\norm{(f-g)\bodisp^{-1/2}}\right\}+\abs{\Ima\sprod{f}{g}}.  
\end{align*}
\end{lemma}
\begin{proof}
Let $\Psi\in D(d\Gamma(\bodisp)^{1/2})$. In analogy to the proof of 
Lemma~B.4 of~\cite{GriesemerWuensch}, we get
\begin{align*}
   \norm{(e^{i\pi(f)}-e^{i\pi(g)})\Psi}\leq\norm{\pi(f-g)\Psi}+\abs{\Ima\sprod{g}{f}}\norm{\Psi}.
\end{align*}
By Lemma~\ref{cor:SpecialEstAnnihilOp}, we obtain
\begin{align*}
   \norm{\pi(f-g)\Psi}\leq 2\max\left\{\norm{f-g},\norm{\frac{f-g}{\bodisp^{1/2}}}\right\}\norm{(1+d\Gamma(\bodisp))^{1/2}\Psi},
\end{align*}
which completes the proof.
\end{proof}

\begin{lemma} 
\label{lm:transformDGamma} 
Let $f\in L^2(\R^d)$ with $\bodisp f\in L^2(\R^d)$. Then 
   $e^{i\pi(f)}D(d\Gamma(\bodisp))=D(d\Gamma(\bodisp))$
and on $D(d\Gamma(\bodisp))$
\begin{align*}
   e^{i\pi(f)}d\Gamma(\bodisp)e^{-i\pi(f)}=d\Gamma(\bodisp)+\phi(\bodisp f)+\norm{\bodisp^{1/2}f}^2.
\end{align*}
\end{lemma}

\begin{proof} 
For $\Psi\in D(d\Gamma(\bodisp))\cap\fock_{0}$ we have 
\begin{align*} 
   e^{i\pi(f)}\Psi=\sum\limits_{n\geq 0}\frac{i^n}{n!}\pi(f)^n\Psi
\end{align*}
where $\pi(f)^n\Psi\in D(d\Gamma(\bodisp))$ because $\bodisp f\in L^2(\R^d)$. 
By a straightforward computation
\begin{align*}
     d\Gamma(\bodisp)i^n\pi(f)^n\Psi =  i^n\pi(f)^n d\Gamma(\bodisp) \Psi &-  n i^{n-1}\pi(f)^{n-1} \phi(\bodisp f) \Psi\\  &+  n(n-1) i^{n-2}\pi(f)^{n-2} \|\bodisp^{1/2} f\|^2 \Psi
\end{align*}
and hence 
$$
    \sum\limits_{n\geq 0}\frac{i^n}{n!}d\Gamma(\bodisp)\pi(f)^n\Psi = e^{i\pi(f)}\left(d\Gamma(\bodisp)-\phi(\bodisp f)+\norm{\bodisp^{1/2}f}^2\right)\Psi.
$$
Since $d\Gamma(\bodisp)$ is a closed operator, it follows that 
$e^{i\pi(f)}\Psi\in D(d\Gamma(\bodisp))$ and that the equation of the 
lemma holds for the chosen $\Psi$. These results now extend to all 
$\Psi\in D(d\Gamma(\bodisp))$ because $D(d\Gamma(\bodisp))\cap\fock_{0}$ 
is an operator core and because of the inequality 
$\norm{d\Gamma(\bodisp)e^{i\pi(f)}\Psi}\leq C_{\bodisp,f}\left(\norm{d\Gamma(\bodisp)\Psi}+\norm{\Psi}\right)$, 
which follow from the transformation formula of the lemma on this core and 
from Lemma~\ref{cor:SpecialEstAnnihilOp}. Since the sign of $f$ was immaterial 
in the above arguments, we have shown that 
$e^{i\pi(\pm f)}D(d\Gamma(\bodisp))\subset D(d\Gamma(\bodisp))$. This implies 
that $e^{i\pi(f)}D(d\Gamma(\bodisp))= D(d\Gamma(\bodisp))$.
\end{proof}

The next lemma is not an immediate corollary of the previous one, because of the relaxed assumption on $f$.
\begin{lemma} 
\label{lm:transformSquRootDGamma} 
Suppose $f\in L^2(\R^d)$ with $\bodisp^{1/2}f\in L^2(\R^d)$. 
Then, $e^{i\pi(f)}D(d\Gamma(\bodisp)^{1/2})=D(d\Gamma(\bodisp)^{1/2})$, 
and for all $\Psi\in D(d\Gamma(\bodisp)^{1/2})$,
\begin{align*} 
   \norm{d\Gamma(\bodisp)^{1/2}e^{i\pi(f)}\Psi} \leq \norm{d\Gamma(\bodisp)^{1/2}\Psi}+  \|\bodisp^{1/2}f\| \|\Psi\|.  
\end{align*}
\end{lemma}

\begin{proof} 
Let $\Psi\in D(d\Gamma(\bodisp))$ and define 
$\bodisp_{\eps}:=\bodisp(1+\eps\bodisp)^{-1}$ for $\eps\geq 0$. Then $\bodisp_\eps f\in L^2(\R^d)$, 
$\bodisp_{\eps}(k)>0$ almost everywhere, and hence, by Lemma~\ref{lm:transformDGamma}, 
\begin{align*}
     \norm{d\Gamma(\bodisp_\eps)^{1/2}e^{i\pi(f)}\Psi}^2 &= \sprod{\Psi}{e^{-i\pi(f)} d\Gamma(\bodisp_\eps)e^{i\pi(f)} \Psi}\\
 &= \sprod{\Psi}{(d\Gamma(\bodisp_\eps) - \phi(\bodisp_\eps f) + \|\bodisp_\eps^{1/2}f\|^2)\Psi}\\
 &\leq  \norm{d\Gamma(\bodisp_\eps)^{1/2}\Psi}^2 + 2\|\Psi\| \|a(\bodisp_\eps f)\Psi\| +   \|\bodisp_\eps^{1/2}f\|^2 \|\Psi\|^2\\
 &\leq \left( \norm{d\Gamma(\bodisp_\eps)^{1/2}\Psi}+  \|\bodisp_\eps^{1/2}f\| \|\Psi\| \right)^2,
\end{align*}
where Lemma~\ref{cor:SpecialEstAnnihilOp} with $\alpha=1/2$ was used in the last inequality.  Letting $\eps\to 0$ the desired inequality is obtained for $\Psi\in D(d\Gamma(\bodisp))$. Since $D(d\Gamma(\bodisp))$ is a form core, this inequality extends to all $\Psi\in D(d\Gamma(\bodisp)^{1/2})$ and then it proves that $e^{i\pi(f)}D(d\Gamma(\bodisp)^{1/2})\subset D(d\Gamma(\bodisp)^{1/2})$ provided that $\|\bodisp^{1/2}f\|<\infty$. Since we may replace $f$ by $-f$ in this proof, the converse inclusion holds as well.
\end{proof}

\begin{corollary} 
\label{cor:transformSquRootDGammaSandwich} 
Let $f\in L^2(\R^d)$ with $\bodisp^{1/2}f\in L^2(\R^d)$. Then 
\begin{align*} 
   \norm{(1+d\Gamma(\bodisp))^{1/2}e^{i\pi(f)}(1+d\Gamma(\bodisp))^{-1/2}} &\leq 1+\norm{\bodisp^{1/2}f}, \\
   \norm{(1+d\Gamma(\bodisp))^{-1/2}e^{i\pi(f)}(1+d\Gamma(\bodisp))^{1/2}} &\leq 1+\norm{\bodisp^{1/2}f}.
\end{align*}
\end{corollary}

\begin{proof}
By a computation very similar to the one in the proof of Lemma~\ref{lm:transformSquRootDGamma}, for all $\Psi\in D(d\Gamma(\bodisp)^{1/2})$,
$$
      \norm{(1+d\Gamma(\bodisp))^{1/2}e^{i\pi(f)}\Psi} \leq \norm{(1+d\Gamma(\bodisp))^{1/2}\Psi}+  \|\bodisp^{1/2}f\| \|\Psi\|.  
$$
This implies the first estimate of the corollary. The second one follows from the first, as it concerns the adjoint operator upon replacing $f$ with $-f$.
\end{proof}

The next two lemmas are needed in Section \ref{sec:nelsonAndGross} and 
they generalize statements (a) and (b) of Lemma~3.1 in \cite{GriesemerWuensch}. 
Recall, that $D(H_0)=D(p^2)\cap D(d\Gamma(\bodisp))$ 
and $D(H_0^{1/2})=D(\abs{p})\cap D(d\Gamma(\bodisp)^{1/2})$.

\begin{lemma}
\label{lm:transformPGen}
Let $p:=-i\nabla_x$ and $F_x(k):=f(k)e^{-ikx}$ with 
$f,\,\bodisp^{1/2}f,\, kf,\, kf\bodisp^{-1/2}\in L^2(\mathbb{R}^d)$. Then $e^{i\pi(F)}D(H_0^{1/2})=D(H_0^{1/2})$ and 
\begin{align*}
   e^{i\pi(F)}pe^{-i\pi(F)}=p-\phi(kF)-\langle f,k f\rangle  \qquad \textrm{on } D(H_0^{1/2}).
\end{align*}
\end{lemma}
%

\begin{proof} 
Following the proof of Lemma~3.1(a) in \cite{GriesemerWuensch}, we define $\DD:= D(H_0)\cap\hilbert_0$, which 
is an operator core and hence a form core of $H_0$. For $\Psi\in\DD$ one 
shows that 
\begin{equation}
\label{eq:pUGen}
       p e^{-i\pi(F)}\Psi = e^{-i\pi(F)}(p-\phi(kF) - \sprod{f}{kf})\Psi
\end{equation}
by expanding $e^{-i\pi(F)}$ in its exponential series. Here we used
that $[\phi(kF),\pi(F)] = 2i \sprod{f}{kf} $, which commutes with $\pi(F)$.
Since $\DD$ is a form core of $H_0$ and since $(p-\phi(kF)-\sprod{f}{kf})$ is bounded 
with respect to $H_0^{1/2}$, by Lemma~\ref{cor:SpecialEstAnnihilOp}, 
Equation~\eqref{eq:pUGen} extends to all 
$\Psi\in D(H_0^{1/2})$ and we see that $e^{-i\pi(F)}D(H_0^{1/2}) \subset D(|p|)$. 
Since $D(H_0^{1/2}) = D(|p|)\cap D(d\Gamma(\bodisp)^{1/2})$ and since 
$D(d\Gamma(\bodisp)^{1/2})$ is left invariant by $e^{-i\pi(F)}$ because 
of the assumption $\bodisp^{1/2}f\in L^2(\mathbb{R}^d)$, see 
Lemma~\ref{lm:transformSquRootDGamma}, we conclude that 
$e^{-i\pi(F)}D(H_0^{1/2})\subset D(H_0^{1/2})$. Likewise, 
$e^{i\pi(F)}D(H_0^{1/2})\subset D(H_0^{1/2})$ by changing the sign of $f$ 
and the lemma is proved.
\end{proof}


The condition $\sprod{f}{kf}=0$ in the following lemma simplifies the identity \eqref{eq:pUGen}, established by Lemma~\ref{lm:transformPGen}, but it is otherwise inessential.

\begin{lemma} 
\label{lm:transformPSquGen}
Let $p:=-i\nabla_x$ and $F_x(k):=f(k)e^{-ikx}$ with 
$f\in L^2(\R^d)$ such that $\bodisp f,k^2f,f\bodisp^{-1/2},k^2f\bodisp^{-1/2}\in L^2(\mathbb{R}^d)$ 
and $\langle f,k f\rangle=0$. Then, $e^{i\pi(F)}D(H_0)=D(H_0)$ and 
\begin{align*}
   e^{i\pi(F)}p^2e^{-i\pi(F)}=\left(p-\phi(kF)\right)^2 \qquad \textrm{on } D(H_0).
\end{align*}
\end{lemma}
\begin{proof} 
Let $\Psi\in D(H_0)\subset D(H_0^{1/2})$. Then, by Lemma~\ref{lm:transformPGen},  
$e^{-i\pi(F)}\Psi\in D(H_0^{1/2})$ and
\begin{align*}
   pe^{-i\pi(F)}\Psi=e^{-i\pi(F)}\left(p-\phi(kF)\right)\Psi.
\end{align*}
Since $\Psi\in D(H_0)$, it is clear that $p\Psi\in D(\abs{p})$. By 
Lemma~\ref{cor:SpecialEstAnnihilOp} and standard estimates, 
\begin{align*}
   \norm{p\phi(kF)} &\leq\norm{\phi(k^2F)\Psi}+\norm{\phi(kF)p\Psi} \\
                    &\leq 2\max\left\{\norm{k^2f},\norm{k^2f\bodisp^{-1/2}}\right\}\norm{(1+d\Gamma(\bodisp))^{1/2}\Psi} \\
                    &\phantom{==} +2\max\left\{\norm{kf},\norm{kf\bodisp^{-1/2}}\right\}\norm{(1+d\Gamma(\bodisp))^{1/2}p\Psi},
\end{align*}
where 
\begin{align}
\label{equ:youngInequalityOp}
   \norm{(1+d\Gamma(\bodisp))^{1/2}p\Psi}\leq\norm{(1+d\Gamma(\bodisp))\Psi}+\norm{p^2\Psi},
\end{align}
which shows that $\phi(kF)\Psi\in D(\abs{p})$ as well, by the assumptions on $f$. 
Thus, $(p-\phi(kF))\Psi\in D(\abs{p})$.

From~\eqref{equ:youngInequalityOp} it also follows that $p\Psi\in D(d\Gamma(\bodisp)^{1/2})$ 
and from Lemma~\ref{lm:ladderMapsOnDomOfSquSecQuant} that
\begin{align*}
   \norm{d\Gamma(\bodisp)^{1/2}\phi(kF)\Psi}\leq 4\max\left\{\norm{\bodisp^{1/2}kf},\norm{kf},\norm{kf\bodisp^{-1/2}}\right\}
                                                \left(\norm{d\Gamma(\bodisp)\Psi}+\norm{\Psi}\right)
\end{align*}
such that $\phi(kF)\Psi\in D(d\Gamma(\bodisp)^{1/2})$ as well, by the assumptions on $f$. 
Therefore, $(p-\phi(kF))\Psi\in D(\abs{p})\cap D(d\Gamma(\bodisp)^{1/2})=D(H_0^{1/2})$, 
such that, using Lemma~\ref{lm:transformPGen} again, 
$e^{-i\pi(F)}(p-\phi(kF))\Psi\in D(H_0^{1/2})$ and we can apply $p$ on this state again. 
This leads to
\begin{align*}
   p^2e^{-i\pi(F)} \Psi = e^{-i\pi(F)}(p-\phi(kF))^2\Psi 
\end{align*}
and means also that $e^{-i\pi(F)} \Psi\in D(p^2)$. From Lemma~\ref{lm:transformDGamma}, we 
already know that $e^{-i\pi(F)} D(H_0)\subset D(d\Gamma(\bodisp))$ such that we conclude 
that $e^{-i\pi(F)} D(H_0)\subset D(H_0)$. Likewise, 
$e^{i\pi(F)}D(H_0)\subset D(H_0)$ by changing the sign of $f$ and the lemma is proved.
\end{proof}




\section*{Acknowledgement}
The authors thank Jacob Schach M{\o}ller for many discussions on this and related projects in the course of the past years. 
A.W.~also had a useful discussion with Stefan Teufel that inspired our proof of Theorem~\ref{thm:domainIntersec3D4D}, and he had many discussions with Martin K{\"o}nenberg on technical issues. The work of Andreas W{\"u}nsch was supported by the \emph{Deutsche Forschungsgemeinschaft (DFG)} through the Research Training Group 1838: \emph{Spectral Theory and Dynamics of Quantum Systems}.



\end{document}